%% file: main.tex
\documentclass[12pt,onecolumn]{IEEEtran}
\usepackage{amsmath,amsfonts}
\usepackage{algorithmic}
\usepackage{algorithm}
\usepackage{array}
\usepackage[caption=false,font=footnotesize]{subfig}
\usepackage{textcomp}
\usepackage{stfloats}
\usepackage{url}
\usepackage{verbatim}
\usepackage{graphicx}
\usepackage{cite}
\usepackage{ieee}
\usepackage{hyperref}

\begin{document}

\title{Distributed Hypothesis Testing Against Dependence}

\author{Han Wu and Shun Watanabe
        \thanks{The authors are with Tokyo University of Agriculture and Technology, Japan. Email: hanwu@go.tuat.ac.jp, shunwata@cc.tuat.ac.jp.
        This work was supported by JST, CRONOS, Japan Grant Number JPMJCS25N5.}
}



\maketitle

\begin{abstract}
We study distributed hypothesis testing and establish the exact error exponent in single-letter form for new testing problems. 
In distributed hypothesis testing, a receiver decides between \(\mathcal{H}_0:P_{XY}\) and \(\mathcal{H}_1:Q_{XY}\) based on \(Y^n\) and a rate-limited description of \(X^n\). 
So far, such single-letter forms are known only for testing against independence, studied by Ahlswede and Csiszár, and testing against conditional independence, studied by Rahman and Wagner.
In this paper, we study testing against dependence, where \(P_{XY}=P_XP_Y\), and show that its error exponent is given by Han’s exponent, which is established by single-letterizing a multi-letter version of Han's exponent. 
We also disprove a previous conjecture of Han claiming that the error exponent is given by the lautum information, as we show that Han's exponent can  strictly exceed the latter.
We then consider the Cartesian product of testing against dependence and testing against independence, and derive a single-letter characterization of its error exponent.
We next study testing against conditional dependence, which is the dependence-testing counterpart of the setting studied by Rahman and Wagner.
We derive a single-letter converse bound for the error exponent by introducing and solving a related setting where the side information is also available to the transmitter.
We show that our converse bound is tight in some cases by using a conditional-coding version of the quantization scheme.
Finally, we extend our converse method for testing against dependence to the general setting, and recover a recently established single-letter converse bound.

\end{abstract}
\section{Introduction}

We study distributed hypothesis testing, formulated as follows.
Consider a pair of jointly i.i.d. sequences \((X^n,Y^n)\).
A transmitter observes \(X^n\) and describes it to a receiver via a noiseless link of capacity \(R\).
Based on the description and \(Y^n\), the receiver attempts to decide between two hypotheses on the distribution of \((X^n, Y^n)\), which we denote by \(\mathcal{H}_0: P_{XY}\) and \(\mathcal{H}_1: Q_{XY}\).
We obtain two error probabilities from the decision process, one for each hypothesis.
Specifically, we obtain the type-I error probability if the receiver decides \(\mathcal{H}_1\) when \(\mathcal{H}_0\) is true, and the type-II error probability if it decides  \(\mathcal{H}_0\) under \(\mathcal{H}_1\).
Following the Neyman–Pearson formulation, we minimize the type-II error probability subject to a prescribed constraint on the type-I error probability. 
In this paper, we focus on the exponential regime, where the performance is measured by the decay rate of the type-II error probability, which is known as the error exponent, denoted by \(E(R)\).
Such testing problems often arise in applications where a central server needs to make decisions from separately collected and encoded data, motivating the characterization of \(E(R)\).

Ahlswede and Csiszár \cite{ahlswedeHypothesisTestingCommunication1986} first established a multi-letter characterization of \(E(R)\).
Since it offers limited practical insight, subsequent work has instead focused on establishing a single-letter characterization, which remains an open problem.
We here briefly review some achievability and converse results relevant to the present work, with additional literature discussed later in this section.
On the achievability side, Ahlswede and Csiszár \cite{ahlswedeHypothesisTestingCommunication1986} derived a single-letter lower bound for \(E(R)\) using the so-called \emph{quantization} scheme: the transmitter quantizes \(X^n\) using a channel \(P_{U|X}\) and sends the quantization output \(U^n\) to the receiver;
the receiver then decides based on \(U^n\) and \(Y^n\).
Han \cite{hanHypothesisTestingMultiterminal1987} improved their bound through a more refined error analysis of the same scheme, which became known as Han's exponent.
Later, Shimokawa, Han, and Amari \cite{shimokawaErrorBoundHypothesis1994} improved the quantization scheme by introducing binning at the transmitter \cite{slepianNoiselessCodingCorrelated1973,wynerRatedistortionFunctionSource1976}.
In their scheme, the transmitter partitions the quantization sequences into bins and sends the bin index to the receiver.
Given the bin index, the receiver attempts to decode the correct sequence from the bin, and uses it for decision.
We call the resulting exponent the SHA exponent.
Compared to Han's exponent, the SHA exponent has an additional constraint imposed by the decoding error event.
However, binning allows us to optimize over a larger set of \(P_{U|X}\), enabling the SHA exponent to exceed the former in some cases.

While the achievability side has seen considerable progress, the converse has seen less.
The key difficulty lies in single-letterizing Ahlswede and Csiszár's multi-letter characterization, which henceforth is referred to as Ahlswede–Csiszár’s exponent.
So far, Ahlswede–Csiszár’s exponent has been successfully single-letterized in only two special testing settings, leading to single-letter characterizations of \(E(R)\) in these cases.
The first is testing against independence, where we have \(Q_{XY} = P_{X}P_{Y}\).
In this setting, Ahlswede-Csiszár's exponent reduces to mutual information, which can be single-letterized using conventional tools from network information theory.
Ahlswede and Csiszár \cite{ahlswedeHypothesisTestingCommunication1986} showed that in this setting we have
\begin{equation}
    E(R) = \max_{\substack{P_{U|X}:\\ I(U;X) \leq R }} I(U;Y).
\end{equation}
The second is testing against conditional independence introduced by Rahman and Wagner \cite{rahmanOptimalityBinningDistributed2012}, where the receiver has access to an additional side information sequence \(Z^n\) and decides between \(\mathcal{H}_0: P_{XYZ}\) and \(\mathcal{H}_1: P_{X|Z}P_{Z}P_{Y|Z}\).\footnote{We only consider a single transmitter here, while Rahman and Wagner \cite{rahmanOptimalityBinningDistributed2012} studied a more general setup with multiple transmitters.}
In this setting, Ahlswede-Csiszár's exponent reduces to conditional mutual information, which can then be single-letterized by identifying appropriate auxiliary random variables.

In this paper, we derive single-letter converse bounds for new testing settings. 
Most of these bounds are achievable using either the quantization scheme itself or modifications of the scheme, without binning.
We therefore characterize \(E(R)\) in single-letter form whenever the achievability and converse bounds match.
For testing against conditional independence, Rahman and Wagner showed that \(E(R)\) is achievable via the SHA exponent, demonstrating that binning can be necessary. 
It remains unclear under what conditions binning is strictly necessary. 
Our results thus uncover several settings in which the optimal exponents can be achieved without it.
Lastly, we show that our converse method extends to the general setting and recover the recently established single-letter converse bound in \cite[Corollary 3]{wenExactRateExponent2026}.
We now introduce the special testing settings considered in this paper.

\subsection{Testing Against Dependence}
We first study the setting where \(P_{XY}= P_{X}P_{Y}\), known as testing against dependence.
We show that in this setting \(E(R)\) is given by Han's exponent.
The achievability part follows from the quantization scheme and Han's analysis, so our contribution is the converse part where we derive the matching converse bound.
Because of the similarity between testing against dependence and testing against independence, Han \cite{Han2012Tradeoff} previously conjectured that, if \(Q_{XY}  = P_{XY}\), then \(E(R)\) in this setting should be given by \(D(P_{U}P_{Y}\|P_{UY})\), which is known as the lautum information \cite{palomarLautumInformation2008}.
In this paper, we will prove that Han's exponent can strictly exceed the lautum information, and disprove his conjecture.

In addition to testing against dependence, we also study two related problems.
The first, which we call \emph{testing against the product of dependence and independence}, is obtained by taking the Cartesian product of the two problems. 
Consider two pairs of jointly i.i.d. sequences, \((X_1^n, Y_1^n)\) and \((X_2^n, Y_2^n)\), which are independent of each other.
The transmitter describes \((X_1^n, X_2^n)\) to the receiver.
Based on the description and \((Y_1^n, Y_2^n)\), the receiver decides between
\begin{equation*}
    \mathcal{H}_0:P_{X_1}P_{Y_1}P_{X_2Y_2} \quad \text{and} \quad \mathcal{H}_1: Q_{X_1Y_1}Q_{X_2}Q_{Y_2}.
\end{equation*}
This problem is motivated by a similar problem studied in the broadcast channel, known as the product of degraded broadcast channels \cite[Problem 5.11]{gamalNetworkInformationTheory2011}.
We show that in this setting \(E(R)\) is also given by Han's exponent.
Since this setting lies outside both the dependence and independence settings, the result shows that the optimality of Han’s exponent extends beyond these two cases.
It also shows that binning can be unnecessary even when the distributions at the transmitter and receiver are correlated under both hypotheses.
We describe the second problem in more detail in the next subsection.

\subsection{Testing Against Conditional Dependence}
The second problem, which we call testing against conditional dependence, is the dependence-testing counterpart of the problem studied by Rahman and Wagner.
Consider a tuple of jointly i.i.d. sequences \((X^n, Y^n, Z^n)\).
The transmitter has access to \(X^n\) and describes it to the receiver.
Based on the description and \((Y^n, Z^n)\), the receiver decides between
\begin{equation*}
    \mathcal{H}_0: P_{X|Z}P_{Z}P_{Y|Z}  \qquad \text{and} \qquad \mathcal{H}_1: Q_{XYZ}.
\end{equation*}
We study this problem by first introducing a related problem in which the side information \(Z^n\) is also available to the transmitter.
That is, the transmitter has access to \((X^n, Z^n)\) and describes them to the receiver, which makes a decision based on the description and \((Y^n, Z^n)\).
We call this related problem \emph{conditional testing against dependence}, owing to its similarity to conditional source coding.
We establish the optimal error exponent in single-letter form for this problem, which generalizes Han's exponent to incorporate the side information \(Z^n\).
In particular, the achievability part is established using a conditional-coding version of the quantization scheme of Ahlswede and Csiszár \cite{ahlswedeHypothesisTestingCommunication1986} and Han \cite{hanHypothesisTestingMultiterminal1987}.
Since the only difference between the two problems is the absence of \(Z^n\) at the transmitter, we readily obtain a single-letter upper bound on \(E(R)\) for testing against conditional dependence.

It remains to see whether our upper bound is tight.
We observe that in general there is a gap between existing achievable exponents and our upper bound.
To demonstrate the tightness of our bound, we consider the case where \(Z\) is a function of \(X\) under both hypotheses. 
The transmitter can then recover \(Z^n\) from \(X^n\), reducing the problem to conditional testing against dependence. 
The upper bound is therefore achievable using the conditional quantization scheme.
Although the same exponent can also be achieved via the SHA scheme, the conditional quantization scheme shows that binning is not needed.

\subsection{Literature Review}
\label{sec:introduction:literature_review}
Before summarizing the contributions of this paper, we briefly review the literature on distributed hypothesis testing.
The problem was first investigated by Berger \cite{Berger1979Decentralized} and Ahlswede and Csiszár \cite{ahlswedeHypothesisTestingCommunication1986}.
The latter derived a multi-letter characterization of \(E(R)\) and also an achievable single-letter lower bound using the quantization scheme.
Han \cite{hanHypothesisTestingMultiterminal1987} then derived a tighter lower bound using the same scheme, and also considered the zero-rate case, which was further studied in \cite{shalabyMultiterminalDetectionZerorate1992}.
Later, Shimokawa \emph{et al.} \cite{shimokawaErrorBoundHypothesis1994} improved the quantization scheme by introducing binning, and their scheme was recently further improved by Kochman and Wang \cite{kochmanImprovedRandomBinningExponent2025}.
Watanabe \cite{watanabeSuboptimalityRandomBinning2022} showed that the quantization-binning scheme is not optimal and also introduced the analysis of the minimal rate that attains Stein's exponent. 
The trade-off between two error exponents was studied in \cite{hanExponentialtypeErrorProbabilities1989} and \cite{weinbergerReliabilityFunctionDistributed2019}.
Achievability regarding Körner-Marton coding was investigated in \cite{haimBinaryDistributedHypothesis2016, GirCunTel26}.
A survey can be found in \cite{hanStatisticalInferenceMultiterminal1998}.
Converse bounds by introducing auxiliary side information were established in  \cite{rahmanOptimalityBinningDistributed2012,kochmanImprovedRandomBinningExponent2025}.
A converse bound by introducing auxiliary receiver was recently established in \cite{wenNewUpperBound2025}.
Multi-terminal and other variants were considered in, e.g., \cite{shalabyErrorExponentsDistributed1994, tianSuccessiveRefinementHypothesis2008,katzDistributedBinaryDetection2017, zhaoDistributedTestingCascaded2018,sreekumarDistributedHypothesisTesting2018,gilaniDistributedHypothesisTesting2019,salehkalaibarHypothesisTestingTwohop2019, weinbergerExponentTradeoffHypothesis2019, sreekumarDistributedHypothesisTesting2020,salehkalaibarDistributedHypothesisTesting2020,salehkalaibarDistributedHypothesisTesting2020a,salehkalaibarDistributedSequentialHypothesis2021,hamadMultihopNetworkMultiple2023, zaidiRateexponentRegionClass2023}.

\textbf{Note:} As we were completing the manuscript, we became aware of the independent work \cite{wenExactRateExponent2026}.
Among other results, the authors of \cite{wenExactRateExponent2026} also solved testing against dependence using a different method, based on the auxiliary receiver approach \cite{wenNewUpperBound2025}.
They also derived a single-letter converse bound for \(E(R)\) under the general setting, see \cite[Corollary 3]{wenExactRateExponent2026}.
Inspired by their converse bound, we extended the converse method developed in this paper to the general setting and recovered the same bound. 
Since our proof provides additional insights to the bound, it  is included in this paper as the final main contribution. We discuss the differences between their proof and ours in the next subsection.

\subsection{Contributions and Organization}
We now summarize the main technical contributions of this paper.
We study testing against dependence, and show that in this setting \(E(R)\) is given by Han's exponent.
We establish the matching converse bound using a converse method that  differs from existing methods \cite{rahmanOptimalityBinningDistributed2012,kochmanImprovedUpperBound2024,wenNewUpperBound2025}.
Existing methods attempt to single-letterize Ahlswede-Csiszár's exponent directly by introducing auxiliary side information or receivers.
In our method, we first upper bound Ahlswede-Csiszár's exponent by a multi-letter version of Han's exponent, which is then single-letterized using ideas inspired by the proof of \cite[Lemma 7]{wuExponentialErrorBounds2026}.
In particular, we perform the single-letterization by manually selecting an auxiliary distribution that yields the desired bound.
To do so, we need to show that it satisfies appropriate marginal constraints, which we establish using a backward summation argument.

We show that for testing against dependence, Han's exponent can strictly exceed the one given by the lautum information, thus disproving Han's conjecture.
We first show that Han's exponent is generally larger than the latter, and then study the conditions under which the two coincide.
In the binary case, we show that they coincide if and only if both equal the time-sharing line between the zero-rate case and the centralized decision case. 
To disprove Han’s conjecture, we provide an example in which both exponents lie strictly above the time-sharing line, illustrating the advantage of Han’s exponent.

We then study testing against the product of dependence and independence, and show that in this setting \(E(R)\) is also given by Han's exponent.
We establish the converse by combining the converse proofs of the individual problems.
We also show that \(E(R)\) has a simpler expression, which can be viewed as the Minkowski sum of the exponents for the two settings.

We next study testing against conditional dependence.
We introduce a related problem, termed conditional testing against dependence, in which the side information is also available to the transmitter.
For this problem, we derive a single-letter characterization of \(E(R)\) by adapting the proof for testing against dependence to the conditional setting.
In particular, the achievability is established using a conditional-coding version of the quantization scheme.
We then leverage the connection between the two problems to derive a single-letter upper bound on \(E(R)\) for testing against conditional dependence.
Further, we identify a special case in which this upper bound is tight by using the conditional quantization scheme.
Together, the optimal exponents in this paper uncover settings in which binning is unnecessary.

Finally, we adapt our converse proof for testing against dependence to the general setting, and recover the same single-letter converse bound recently established in \cite[Corollary 3]{wenExactRateExponent2026}.
Our proof for the general setting largely follows that for testing against dependence, but requires a different choice of auxiliary random variable and, consequently, a more sophisticated backward summation argument. 
Since \cite[Corollary 3]{wenExactRateExponent2026} is proved using the auxiliary receiver approach \cite{wenNewUpperBound2025}, some differences between their proof and ours have already been discussed in the context of testing against dependence.
Here, we highlight two additional ones.
First, key to the single-letterization method in \cite{wenExactRateExponent2026} is a telescoping identity argument \cite{kramerTeachingITIdentity2011,gohariOuterBoundsMultiuser2022}, which does not appear in ours.
Second, they use a coupling argument, based on the data processing inequality, to introduce the single-letter marginal constraints in Han’s exponent.
In contrast, we first use the data processing inequality to obtain a multi-letter version of Han’s exponent, whose marginal constraints then require the single-letter ones.

The rest of the paper is organized as follows. 
After describing key notations at the end of this section, in the next section we provide a formal description of testing against dependence.
In Section \ref{sec:main_results}, we provide the main results of this paper and provide insights, where we will also describe the two related problems in further detail.
Sections \ref{sec:converse_proof_thm_testing_against_dependence} to \ref{apd:upper_bound_general_setting} are dedicated to proving the main results, while proofs of some technical lemmas are deferred to the appendices. 
Concluding remarks and future directions are provided in Section \ref{sec:concluding_remarks}, and AI use in this paper is disclosed in the acknowledgment section.

\subsection{Notation}
All alphabets in this paper are finite.
We use \(\mathcal{P}(\mathcal{X})\) to denote the set of all pmfs \(P_X\) on \(\mathcal{X}\).
We write \(\bm{x}=(x_1, x_2, \ldots, x_n)\) for an \(n\)-length sequence from \(\mathcal{X}^n\).
Depending on the context, we may also write \(x^n\) and \(X^n\) instead of \(\bm{x}\) and \(\bm{X}\).
In the same manner, we adopt the notation \(\bm{y} = (y_1, y_2, \ldots, y_n)\) or \(\bm{u}=(u_1, u_2, \ldots, u_n)\), and \(\bm{Y}\) or \(\bm{U}\), on \(\mathcal{Y}^n\) or \(\mathcal{U}^n\) respectively.
Following convention, the hat symbol \(\hat{P}\) is used whenever we are looking at the empirical distribution induced by some deterministic sequences.
For example, for a sequence \(\bm{x} \in \mathcal{X}^{n}\), we use \(\hat{P}_{\bm{x} }\) to denote its vector of relative frequencies of all symbols \(x \in \mathcal{X}\), i.e., its type.
\(\hat{P}_{\bm{x}\bm{y}}\) denotes the joint type of a sequence pair \((\bm{x}, \bm{y})\), while \(\hat{P}_{\bm{y} | \bm{x}}\) is the conditional type from \(\bm{x}\) to \(\bm{y}\) induced by  \(\hat{P}_{\bm{x}\bm{y}}\).
The set of all possible types \(\hat{P}_{\bm{x}}\) on \(\mathcal{X}^{n}\) is written as \(\mathcal{P}_n(\mathcal{X})\), while the set of all possible conditional types \(\hat{P}_{\bm{y} | \bm{x}}\) for sequences from \(\mathcal{X}^n\) and \(\mathcal{Y}^n\) is written as \(\mathcal{P}_n(\mathcal{Y} | \mathcal{X})\).
The type class \(\mathcal{T}_n(P_X)\) consists of all sequences \(\bm{x}\) that have the same type \(P_X \in \mathcal{P}_n(\mathcal{X})\).
For a given sequence \(\bm{x}\), the conditional type class \(\mathcal{T}_n(P_{Y|X} | \bm{x})\) is the set of all sequences \(\bm{y}\) such that the conditional type from \(\bm{x}\) to \(\bm{y}\) is \(P_{Y|X} \in \mathcal{P}_n(\mathcal{Y} | \mathcal{X})\).
The set of all types and conditional types are written as \(\mathcal{P}_n(\mathcal{X})\) and \(\mathcal{P}_n(\mathcal{Y}|\mathcal{X})\), with a type class being \(\mathcal{T}_n(Q_X)\) and a conditional type class being \(\mathcal{T}_n(Q_{Y|X} | \bm{x})\).

The entropy of \(P_{X}\) is written as \(H(X)\) or \(H(P_X)\) and the conditional entropy between two random variables \(X\) and \(Y\) is denoted by \(H(Y|X)\) or \(H(P_{Y|X} | P_X)\), while the mutual information between \(X\) and \(Y\) is written as \(I(X;Y)\) or \(I(P_X, P_{Y|X})\).
\(D(Q_X\|P_X)\) is the KL-divergence between two pmfs \(Q_X\) and \(P_X\), and \(D(Q_{Y|X}\|P_{Y|X} | P_X)\) denotes the conditional KL-divergence.
Given an event \(\mathcal{A}\), we use \(P[\mathcal{A}]\) to denote the probability of \(\mathcal{A}\) under the probability measure \(P\), while \(\idc{\mathcal{A}}\) is the indicator function of \(\mathcal{A}\) and \(\abs{\mathcal{A}}\) is its cardinality or size.
Let \(a_n \ndot{\leq} b_n\) if \(\limsup_{n \to \infty}\frac{1}{n}\log (a_n/b_n) \leq 0\) and \(a_n \ndot{\geq} b_n\) if \(\liminf_{n \to \infty}\frac{1}{n} \log (a_n/b_n) \geq 0\).
We will write \(a_n \ndot{=} b_n\) if both hold.
For a positive integer constant \(N\), we use \([N]\) to denote \(\{1,2,\ldots,N\}\).
Let \(|a|^{+} \triangleq \max\{0,a\}\).
The base of exponential and log functions is chosen as the natural base.

\section{Problem Setup}
\label{sec:problem_setup}
We first describe the problem setting in further detail.
Consider an i.i.d. sequence pair \((X^n, Y^n)\), generated according to one of the following hypotheses:
\begin{align}
    \mathcal{H}_0: (X, Y) & \sim P_{XY}, \\
    \mathcal{H}_1: (X, Y) & \sim Q_{XY}.
\end{align}
We assume \(D(P_{XY} \| Q_{XY}) < \infty\) throughout this paper.
As illustrated in Fig. \ref{fig:problem-setup:DHT-model}, a transmitter observes \(X^n\) and describes it to a receiver using an encoder \(f_n:\mathcal{X}^n \to \mathcal{M}_n \).
Let \(M = f_n(X^n)\) denote the message forwarded by the transmitter.
The receiver receives \(M\) and observes \(Y^n\); then decides between \(\mathcal{H}_0\) and \(\mathcal{H}_1\) based on a test \(\varphi_n: \mathcal{Y}^n \times \mathcal{M}_n \to \{\mathcal{H}_0, \mathcal{H}_1\}\).

\begin{figure}[hbt!]
    \centering
    \scalebox{0.88}{\input{fig/DHT_model.tex}}
    \captionsetup{justification=centering}
    \caption{Distributed Hypothesis Testing}
    \label{fig:problem-setup:DHT-model}
\end{figure}

Given a testing scheme \((f_n, \varphi_n)\), two error probabilities arise, corresponding to incorrect decisions under each hypothesis:
\begin{align}
    \alpha_n & \triangleq  P_{XY}^n[ \varphi_n( Y^n, f_n(X^n)  ) = \mathcal{H}_1  ], \\
    \beta_n & \triangleq  Q_{XY}^n[ \varphi_n( Y^n, f_n(X^n)  ) = \mathcal{H}_0  ].
\end{align}
We follow the Neyman-Pearson formulation for hypothesis testing, where the goal is to minimize the type-II error probability \(\beta_n\) subject to a constraint on \(\alpha_n\).
We say a rate-exponent pair \((R, E)\) is achievable if there exists a sequence of testing schemes \((f_n, \varphi_n)\) satisfying \( ||f_n|| \ndot{\leq} e^{nR}\) such that
\begin{align}
    \lim_{n\to \infty} \alpha_n & = 0,\\
    \liminf_{n\to \infty} -\frac{1}{n}\log\beta_n & \geq E.
\end{align}
For a fixed \(R\), we are interested in the maximum achievable exponent
\begin{equation*}
    E(R) \triangleq \max \{E: (R, E) \text{ is achievable} \}.
\end{equation*}
Ahlswede and Csiszár \cite{ahlswedeHypothesisTestingCommunication1986} established a multi-letter characterization of \(E(R)\), given by
\begin{equation}
    E(R) = \lim_{\tau \to 0}\sup_{n}\sup_{ \substack{ f_n : \\ ||f_n|| \leq e^{n(R+\tau)} } } \frac{1}{n} D(P_{MY^n} \| Q_{MY^n}). \label{eq:problem_setup:exponent_multi_letter_characterization}
\end{equation}
We henceforth refer to the RHS of \eqref{eq:problem_setup:exponent_multi_letter_characterization} as Ahlswede-Csiszár's exponent.
In this paper, we focus on deriving a single-letter characterization of \(E(R)\) for a special testing setting, known as testing against dependence. 
We then study variants of this setting and finally return to the general setting by extending the converse proof. 
Before introducing this setting, we first formally define Han’s exponent for distributed hypothesis testing, which is central to our results.

\subsection{Han's Exponent}
We first introduce some notation.
Given any \(P_{U|X}\), let \(P_{UXY} = P_{U|X}P_{XY}\), \(Q_{UXY} = P_{U|X}Q_{XY}\), and define
\begin{equation}
    E(P_{UXY} \| Q_{UXY}) \triangleq \min_{\substack{\tilde{P}_{UXY} \in \mathcal{P}_{\mathrm{H}}(P_{UXY})  }} D( \tilde{P}_{UXY} \| Q_{UXY} ), \label{eq:problem_setup:Han_exponent_definition_1}
\end{equation}
where 
\begin{equation}
     \mathcal{P}_{\mathrm{H}}(P_{UXY})  \triangleq \{ \tilde{P}_{UXY}: \tilde{P}_{UX} = P_{UX}, \tilde{P}_{UY} = P_{UY}\}. \label{eq:problem_setup:Han_exponent_definition_2}
\end{equation}
Han \cite{hanHypothesisTestingMultiterminal1987} showed that \(E(R) \geq E_{\mathrm{H}}(R)\), where
\begin{align}
    E_{\mathrm{H}}(R) & \triangleq \max_{\substack{P_{U|X}: \\ I(U;X) \leq R } } E(P_{UXY} \| Q_{UXY}), \label{eq:problem_setup:Han_exponent}
\end{align}
in which \(I(U;X) = I(P_{X}, P_{U|X})\) and \(|\mathcal{U}| \leq |\mathcal{X}| + 1\).
We henceforth refer to \(E_{\mathrm{H}}(R)\) as Han's exponent.

\subsection{Testing Against Dependence}
We now describe the testing setting.
In this paper, we study the following class of distributed hypothesis testing problems, where the receiver decides between
\begin{align}
    \mathcal{H}_0: (X, Y) & \sim P_{X}P_{Y}, \\
    \mathcal{H}_1: (X, Y) & \sim Q_{XY}.
\end{align}
That is, the first hypothesis is a product distribution.
Note that here the two hypotheses can have different marginal distributions, different from testing against independence studied in the literature.
Due to the similarity between the independence-testing and dependence-testing settings, Han \cite{Han2012Tradeoff} conjectured that for testing against dependence with \(Q_{XY}= P_{XY}\), we have
\begin{equation}
    E(R) = \max_{\substack{P_{U|X}: \\ I(U;X) \leq R}} L(U;Y), \label{eq:problem_setup:testing_against_dependence:lautuam_information_exponent}
\end{equation}
where \(P_{UXY} = P_{U|X}P_{XY}\) and \(L(U;Y) \triangleq D(P_{U}P_{Y} \| P_{UY})\), which is known as the lautum information \cite{palomarLautumInformation2008}.
In this paper, we show that in this setting \(E(R)\) is in fact given by Han's exponent \(E_{\mathrm{H}}(R)\) and disprove his conjecture.
We also study two problems related to testing against dependence, which will be formally introduced in the next section.

\section{Main Results and Discussions}
\label{sec:main_results}
\subsection{Testing Against Dependence}
\label{sec:main_results:testing_against_dependence}
We first show that Han's exponent is optimal for testing against dependence. 
To this end, recall the definition of \( E(P_{UXY} \| Q_{UXY})\) in \eqref{eq:problem_setup:Han_exponent_definition_1} and \eqref{eq:problem_setup:Han_exponent_definition_2}.
\begin{theorem}
    \label{thm:testing_against_dependence}
    For testing against dependence, where
    \begin{align}
        \mathcal{H}_0: (X, Y) & \sim P_{X}P_{Y}, \\
        \mathcal{H}_1: (X, Y) & \sim Q_{XY},
    \end{align}
    we have \(E(R) = E_{\mathrm{H}}(R)\), i.e.,
    \begin{align}
        E(R) & = \max_{\substack{P_{U|X}: \\ I(U;X) \leq R } } E(P_{UX}P_{Y} \| Q_{UXY}), \label{eq:main_results:testing_against_dependence:exponent}
    \end{align}
    where \(|\mathcal{U}| \leq |\mathcal{X}| + 1\).
\end{theorem}
\begin{proof}
    The achievability part \(E(R) \geq E_{\mathrm{H}}(R)\) follows from \cite[Theorem 2]{hanHypothesisTestingMultiterminal1987}. We establish the converse part \(E(R) \leq E_{\mathrm{H}}(R)\) in Section \ref{sec:converse_proof_thm_testing_against_dependence}.
\end{proof}

We next consider Han's conjecture and study the case where we have \(\mathcal{H}_0: P_{X}P_{Y}\) and \(\mathcal{H}_1: P_{XY}\).
We show that, in this case, \(E_{\mathrm{H}}(R)\) can strictly exceed the exponent given by the lautum information, thus disproving his conjecture.
For this purpose, let us denote by \(L(R)\) the RHS of \eqref{eq:problem_setup:testing_against_dependence:lautuam_information_exponent}, in which we impose the constraint \(|\mathcal{U}| \leq |\mathcal{X}| + 1\) without loss of generality.
\begin{proposition}
\label{prop:disproving_Han_conjecture}
    For testing against dependence where \(\mathcal{H}_0: (X, Y) \sim P_{X}P_{Y}\) and \(\mathcal{H}_1: (X, Y) \sim P_{XY}\), if \(|\mathcal{X}| = |\mathcal{Y}| = 2\) and  \(R \leq H(X)\), then \(E_{\mathrm{H}}(R) = L(R)\)  if and only if
    \begin{equation}
        E_{\mathrm{H}}(R) = L(R) = \frac{R}{H(X)} D(P_{X}P_{Y} \| P_{XY}).
    \end{equation}
    Consequently, there exists some \(P_{XY}\) such that \(E_{\mathrm{H}}(R) > L(R)\) for every \(0<R<H(X)\).
\end{proposition}
\begin{proof}
    This proposition states that, in the binary case, \(E_{\mathrm{H}}(R) = L(R)\) if and only if both equal the time-sharing line between \((0,0)\) and \(\big(H(X), D(P_{X}P_{Y} \| P_{XY}) \big)\).
    We then provide an example where \(L(R) >  \frac{R}{H(X)} D(P_{X}P_{Y} \| P_{XY})\) for every \(0<R<H(X)\), which implies  \(E_{\mathrm{H}}(R) > L(R)\) for every \(0<R<H(X)\). 
    The detailed proof is provided in Appendix \ref{apd:proof_disproving_Han_conjecture}.
\end{proof}

We now discuss the converse proof of Theorem \ref{thm:testing_against_dependence}.
For simplicity, we focus on the special case \(Q_{XY} = P_{XY}\) in the following, where the two hypotheses have the same marginals.
Our proof begins with Ahlswede-Csiszár's multi-letter characterization of \(E(R)\)  in \eqref{eq:problem_setup:exponent_multi_letter_characterization}, i.e.,
\begin{equation}
    E(R) \leq  \liminf_{n\to \infty}\frac{1}{n} D(P_{M}P_{Y^n} \| P_{MY^n}). \label{eq:main_results:testing_against_dependence:multi_letter_upper_bound}
\end{equation}
The RHS of \eqref{eq:main_results:testing_against_dependence:multi_letter_upper_bound} is known as the lautum information \cite{palomarLautumInformation2008}, rather than the mutual information appearing in testing against independence.
Consequently, the conventional approach used in the latter problem does not directly apply here, thus motivating a different approach.
In our approach, we first upper bound Ahlswede-Csiszár's exponent by a multi-letter version of Han's exponent, i.e.,
\begin{equation}
    D(P_{M}P_{Y^n} \| P_{MY^n}) \leq \min_{\substack{\tilde{P}_{MX^nY^n}}} D(\tilde{P}_{MX^nY^n} \| P_{MX^nY^n}), \label{eq:main_results:testing_against_dependence:data_processing_inequality}
\end{equation}
where the minimization is taken over all \(\tilde{P}_{MX^nY^n}\) satisfying \(\tilde{P}_{MX^n} = P_{MX^n}\) and \(\tilde{P}_{MY^n} = P_{M}P_{Y^n}\).
Due to the minimization, any \(\tilde{P}_{MX^nY^n}\) satisfying the two marginal constraints provides a valid upper bound for the RHS of \eqref{eq:main_results:testing_against_dependence:data_processing_inequality}.
This freedom in choosing \(\tilde{P}_{MX^nY^n}\) allows us to select the one defined by
\begin{equation}
    \tilde{P}_{MX^nY^n} = P_{MX^n}\prod_{i=1}^{n} \tilde{P}_{Y_i | M X^{i-1}X_i}. \label{eq:main_results:testing_against_dependence:auxiliary_distribution}
\end{equation}
However, a key challenge is to ensure that the resulting \(\tilde{P}_{MX^nY^n}\) satisfies \(\tilde{P}_{MY^n} = P_{M}P_{Y^n}\).
We prove this using a backward summation argument.
With this choice of \(\tilde{P}_{MX^nY^n}\), we can then derive the desired single-letter bound and establish the converse.
Note that we will extend the converse proof of Theorem \ref{thm:testing_against_dependence} to the general setting and recover a recently established single-letter converse bound, see Theorem \ref{thm:upper_bound_general_setting}.

\begin{remark}
\label{rem:generalization_testing_against_independence}
    Notice that the two hypotheses in Theorem \ref{thm:testing_against_dependence} can have different marginal distributions. 
    Motivated by this, we consider a more general formulation of testing against independence, where the receiver decides between
    \begin{align}
        \mathcal{H}_0: (X, Y) & \sim P_{XY}, \\
        \mathcal{H}_1: (X, Y) & \sim Q_{X}Q_{Y}.
    \end{align}
    In Appendix \ref{apd:converse_proof_generalization_testing_against_independence},  we point out that we also have \(E(R) = E_{\mathrm{H}}(R)\) under this more general formulation, i.e.,
    \begin{align}
        E(R) & = \max_{\substack{P_{U|X}: \\ I(U;X) \leq R } } E(P_{UXY} \| Q_{UX}Q_{Y}), \label{eq:main_results:testing_against_independence:exponent}
    \end{align}
    where \(|\mathcal{U}| \leq |\mathcal{X}| + 1\).
    Together with Theorem \ref{thm:testing_against_dependence}, this shows that Han's exponent is optimal for distributed hypothesis testing when either hypothesis is a product distribution.
\end{remark}

\subsection{Testing Against the Product of Dependence and Independence}
We next study another distributed hypothesis testing problem, obtained by taking the Cartesian product of testing against dependence and testing against independence.
As illustrated in Fig. \ref{fig:main_results:testing_against_product_dependence_independence}, we follow the same setup of Section \ref{sec:problem_setup}, except that now \(X=(X_1, X_2)\), \(Y=(Y_1, Y_2)\), and the receiver decides between:
\begin{align}
    \mathcal{H}_0: (X_1, X_2, Y_1, Y_2) & \sim P_{X_1}P_{Y_1}P_{X_2Y_2}, \\
    \mathcal{H}_1: (X_1, X_2, Y_1, Y_2) & \sim Q_{X_1Y_1}Q_{X_2}Q_{Y_2}.
\end{align}
In the following result, we establish \(E(R) = E_{\mathrm{H}}(R)\) for this problem.
We also show that \(E(R)\) here has a simpler expression, which can be viewed as the Minkowski sum of the two exponents in \eqref{eq:main_results:testing_against_dependence:exponent} and \eqref{eq:main_results:testing_against_independence:exponent}.

\begin{figure}[t]
    \centering
    \scalebox{0.88}{\input{fig/Testing_against_product_dependence_independence.tex}}
    \captionsetup{justification=centering}
    \caption{Testing Against the Product of Dependence and Independence}
    \label{fig:main_results:testing_against_product_dependence_independence}
\end{figure}

\begin{theorem}
    \label{thm:testing_against_product_dependence_independence}
    For testing against the product of dependence and independence, where 
    \begin{align}
        \mathcal{H}_0: (X_1, X_2, Y_1, Y_2) & \sim P_{X_1}P_{Y_1}P_{X_2Y_2}, \\
        \mathcal{H}_1: (X_1, X_2, Y_1, Y_2) & \sim Q_{X_1Y_1}Q_{X_2}Q_{Y_2},
    \end{align}
    we have \(E(R) = E_{\mathrm{H}}(R)\).
    Moreover, \(E(R)\) here simplifies to
    \begin{align}
        E(R) &= \max_{\substack{P_{U_1|X_1}, P_{U_2|X_2}: \\ I(U_1;X_1) + I(U_2;X_2) \leq R}} E(P_{U_1X_1}P_{Y_1} \| Q_{U_1X_1Y_1}) +  E(P_{U_2X_2Y_2} \| Q_{U_2X_2}Q_{Y_2}), \label{eq:main_results:testing_against_product_dependence_independence:exponent}
    \end{align}
    where \(|\mathcal{U}_1| \leq |\mathcal{X}_1| + 1\) and \(|\mathcal{U}_2| \leq |\mathcal{X}_2| + 1\).
\end{theorem}
\begin{proof}
    See Section \ref{sec:proof_testing_against_product_dependence_independence}.    
\end{proof}

Note that testing against the product of dependence and independence does not fall into either the dependence-testing or the independence-testing setting.
Hence, Theorem \ref{thm:testing_against_product_dependence_independence} shows that the optimality of Han's exponent holds beyond these two settings.
It also shows that binning can be unnecessary even when the transmitter and receiver observations are correlated under both hypotheses.

\subsection{Testing Against Conditional Dependence}
\label{sec:main_results:testing_against_conditional_dependence}
We next study the dependence-testing counterpart of the setting introduced by Rahman and Wagner \cite{rahmanOptimalityBinningDistributed2012}.
We begin by describing another testing problem, which we call \emph{conditional testing against dependence}. 
Its connection to the setting under investigation will become clear shortly.
As illustrated in Fig. \ref{fig:main_results:contional_testing_against_dependence}, the key difference between this problem and testing against dependence is the availability of an additional side information sequence to both the transmitter and the receiver, denoted by \(Z^n\).
The transmitter observes both \(X^n\) and \(Z^n\), and describes them to a receiver using an encoder \(\tilde{f}_n:\mathcal{X}^n \times \mathcal{Z}^n \to \mathcal{M}_n \).
The receiver receives \(M\) and observes \(Y^n\) and \(Z^n\); then decides between \(\mathcal{H}_0\) and \(\mathcal{H}_1\) based on a test \(\varphi_n: \mathcal{Y}^n \times \mathcal{Z}^n \times \mathcal{M}_n \to \{\mathcal{H}_0, \mathcal{H}_1\}\), where
\begin{align}
    \mathcal{H}_0: (X, Y, Z) & \sim P_{X|Z}P_{Y|Z}P_{Z}, \\
    \mathcal{H}_1: (X, Y, Z) & \sim Q_{XY|Z}Q_{Z}.
\end{align}

\begin{figure}[t]
    \centering
    \subfloat[Conditional Testing Against Dependence]{%
        \scalebox{0.88}{\input{fig/Conditional_testing_against_dependence.tex}}
        \label{fig:main_results:contional_testing_against_dependence}
    }
    \hfill
    \subfloat[Testing Against Conditional Dependence]{%
        \scalebox{0.88}{\input{fig/Testing_against_conditional_dependence}}
        \label{fig:main_results:testing_against_conditional_dependence}
    }
    \caption{Testing with Additional Side Information}
\end{figure}



Let \(\tilde{E}(R)\) denote the maximum achievable exponent of the type-II error probability under this testing problem.
In the following result, we establish \(\tilde{E}(R)\) in single-letter form, which turns out to coincide with Han's exponent despite the presence of side information.
Before describing the result, we first introduce some notation.
Consider two pmfs \(P_{XYZ}\) and \(Q_{XYZ}\).
Given any \(P_{U|XZ}\), let \(P_{UXYZ} = P_{U|XZ}P_{XYZ}\), \(Q_{UXYZ} = P_{U|XZ}Q_{XYZ}\), and define
\begin{equation}
    \tilde{E}(P_{UXYZ} \| Q_{UXYZ}) \triangleq \min_{\substack{\tilde{P}_{UXYZ} \in \tilde{\mathcal{P}}_{\mathrm{H}}(P_{UXYZ})  }} D( \tilde{P}_{UXYZ} \| Q_{UXYZ} ),  \label{eq:main_results:conditional_testing_against_dependence:definition_E_tilde}
\end{equation}
where 
\begin{equation}
     \tilde{\mathcal{P}}_{\mathrm{H}}(P_{UXYZ})  \triangleq \{ \tilde{P}_{UXYZ}: \tilde{P}_{UXZ} = P_{UXZ}, \tilde{P}_{UYZ} = P_{UYZ}\}.
\end{equation}

\begin{theorem}
    \label{thm:conditional_testing_against_dependence}
    For conditional testing against dependence, where
    \begin{align}
        \mathcal{H}_0: (X, Y, Z) & \sim P_{XZ}P_{Y|Z}, \\
        \mathcal{H}_1: (X, Y, Z) & \sim Q_{XYZ},
    \end{align}
    we have
    \begin{equation}
        \tilde{E}(R) = \max_{\substack{P_{U|XZ}: \\ I(U;X|Z) \leq R }} \tilde{E}(P_{UXZ}P_{Y|Z} \| Q_{UXYZ}), \label{eq:main_results:conditional_testing_against_dependence:exponent}
    \end{equation}
    where \( |\mathcal{U}| \leq |\mathcal{X}||\mathcal{Z}| + 1 \).
\end{theorem}
\begin{proof}
    We establish the achievability using a conditional-coding version of the quantization scheme, and the converse by adapting the converse proof of Theorem \ref{thm:testing_against_dependence}.
    See Section \ref{sec:achievability_proof_conditional_testing_against_depedence} and Section \ref{sec:converse_proof_conditional_testing_against_depedence}, respectively.
    Note that the achievability proof in Section \ref{sec:achievability_proof_conditional_testing_against_depedence} applies to any joint pmf \(P_{XYZ}\), not only when \(P_{XYZ} = P_{XZ}P_{Y|Z}\).
    In the converse proof, we also prove the cardinality bound \( |\mathcal{U}| \leq |\mathcal{X}||\mathcal{Z}| + 1 \) using a different argument than Han's.
\end{proof}

\begin{remark}
It turns out that the RHS of \eqref{eq:main_results:conditional_testing_against_dependence:exponent} coincides with the exponent obtained by applying Han's exponent \(E_{\mathrm{H}}(R)\) to this problem.
Recall the definition of \(E_{\mathrm{H}}(R)\) in \eqref{eq:problem_setup:Han_exponent}.
By incorporating the side information \(Z^n\) into the observations \(X^n\) and \(Y^n\) at the transmitter and receiver, i.e., by replacing \(X^n\) and \(Y^n\) with \((X^n,Z^n)\) and \((Y^n,Z^n)\) respectively, we can verify that for conditional testing against dependence, \(E_{\mathrm{H}}(R)\) is given by
\begin{equation}
    E_{\mathrm{H}}(R)  = \max_{\substack{P_{V|XZ}: \\ I(V;X, Z) \leq R }} \tilde{E}(P_{VXZ}P_{Y|Z} \| Q_{VXYZ}), \label{eq:main_results:conditional_testing_against_dependence:Han_exponent_under_this_setting}
\end{equation}
where we use \(V\), rather than \(U\), for the auxiliary random variable to distinguish the two expressions.
At first glance, the RHS of \eqref{eq:main_results:conditional_testing_against_dependence:exponent} should be at least as large as that of \eqref{eq:main_results:conditional_testing_against_dependence:Han_exponent_under_this_setting}, since \(I(U;X|Z) \leq I(V;X,Z)\) where \(V=U\).
Yet, the two expressions turn out to be equal.
We prove the equality in Appendix \ref{apd:proof_equality_conditional_and_Han_two_exponents} by enlarging the alphabet of \(V\).
In summary, the original quantization scheme captures the common side information implicitly through the optimization over auxiliary random variables.
In contrast, the conditional quantization scheme makes the role of the common side information more explicit, yielding a more transparent expression of the achievable exponent through the constraint \(I(U;X|Z)\le R\).
\end{remark}

We now return to the initial setting of interest, which we call \emph{testing against conditional dependence}.
As illustrated in Fig. \ref{fig:main_results:testing_against_conditional_dependence}, the only difference between this problem and the problem studied in Theorem \ref{thm:conditional_testing_against_dependence} is that  \(Z^n\) is no longer available to the transmitter.
This immediately gives the following result.


\begin{theorem}
    \label{thm:testing_against_conditional_dependence}
    For testing against conditional dependence, where
    \begin{align}
        \mathcal{H}_0: (X, Y, Z) & \sim P_{XZ}P_{Y|Z}, \\
        \mathcal{H}_1: (X, Y, Z) & \sim Q_{XYZ},
    \end{align}
    we have
    \begin{equation}
        E(R) \leq \tilde{E}(R),
    \end{equation}
    where \(\tilde{E}(R)\) is given by \eqref{eq:main_results:conditional_testing_against_dependence:exponent}.
\end{theorem}

Due to the absence of \(Z^n\) at the transmitter, we see that \(\tilde{E}(R)\) is in general only an upper bound, and a gap may remain between this bound and the existing achievable exponents for this problem. It is therefore natural to ask when the bound is tight.
We next identify an example in which \(\tilde{E}(R)\) is achievable.

\begin{corollary}
\label{cor:main_results:testing_against_conditional_dependence:mapping_case}
    In testing against conditional dependence, if there exists some function \(g\) such that \(Z = g(X)\) holds under both \(\mathcal{H}_0\) and \(\mathcal{H}_1\), then we have \(E(R) = \tilde{E}(R)\).
\end{corollary}
\begin{proof}
    Since \(Z^n = g(X^n)\), the transmitter effectively has access to the side information. Hence, we can establish \(E(R) \geq \tilde{E}(R)\) using the conditional quantization scheme in Section \ref{sec:achievability_proof_conditional_testing_against_depedence}.
\end{proof}

Corollary \ref{cor:main_results:testing_against_conditional_dependence:mapping_case} identifies a setting in which conditional quantization can be used without explicitly providing \(Z^n\) to the transmitter.
When \(Z=g(X)\), the transmitter can reconstruct the receiver's side information from its own observation and hence directly perform quantization conditioned on \(Z^n\).
Although the same exponent in Corollary \ref{cor:main_results:testing_against_conditional_dependence:mapping_case} can be shown to be achievable using the SHA scheme, the conditional quantization scheme demonstrates that binning is not needed.


\subsection{General Setting}
\label{sec:main_results:general_setting}
In our final result, we return to the general setting of distributed hypothesis testing.
We show that our converse proof for testing against dependence in Theorem \ref{thm:testing_against_dependence} extends to the general setting, recovering the same single-letter converse bound recently established in \cite[Corollary 3]{wenExactRateExponent2026}.
The upper bound takes the form of Han’s exponent, but with the maximization now over \(P_{U|X}\) satisfying \(I(U;X|Y)\leq R\), rather than \(I(U;X)\leq R\) as in Han’s exponent, which echoes the conditional rate constraint that emerged earlier in our analysis of conditional testing against dependence.

\begin{theorem}[{\cite[Corollary 3]{wenExactRateExponent2026}}]
    \label{thm:upper_bound_general_setting}
    For distributed hypothesis testing, where
    \begin{align}
        \mathcal{H}_0: (X, Y) & \sim P_{XY}, \\
        \mathcal{H}_1: (X, Y) & \sim Q_{XY},
    \end{align}
    we have
    \begin{align}
        E(R) & \leq \max_{\substack{P_{U|X}: \\ I(U;X|Y) \leq R } } E(P_{UXY} \| Q_{UXY}), \label{eq:main_results:general_setting:exponent}
    \end{align}
    where \(|\mathcal{U}| \leq |\mathcal{X}| + 1\).
\end{theorem}
\begin{proof}
    The proof largely follows that of Theorem \ref{thm:testing_against_dependence} for testing against dependence.
    However, to handle the general setting, we need to choose a different auxiliary random variable, particularly one that includes \(Y^{n}_{i+1}\). 
    As a result, the proof requires a more sophisticated backward summation argument. We provide the detailed proof in Section \ref{apd:upper_bound_general_setting}.
\end{proof}
\begin{remark}
The authors of \cite{wenExactRateExponent2026} concurrently and independently proved Theorem \ref{thm:testing_against_dependence} for testing against dependence in \cite[Lemma 3]{wenExactRateExponent2026}.
Inspired by their result in \cite[Corollary 3]{wenExactRateExponent2026}, we subsequently extended our converse method to the general setting, leading to Theorem \ref{thm:upper_bound_general_setting} of this paper.
We then realized that Theorem \ref{thm:upper_bound_general_setting} can also be applied to testing against conditional dependence to obtain a tightened upper bound over that in Theorem \ref{thm:testing_against_conditional_dependence}, which can be seen by replacing \(Y\) with \((Y,Z)\) in Theorem \ref{thm:upper_bound_general_setting}.
We do not include this additional result here, since Theorem \ref{thm:testing_against_conditional_dependence} had already been completed before we became aware of \cite{wenExactRateExponent2026}.
\end{remark}

\section{Converse Proof of Theorem \ref{thm:testing_against_dependence}}
\label{sec:converse_proof_thm_testing_against_dependence}

We prove the converse part of Theorem \ref{thm:testing_against_dependence} in this section.
Consider a sequence of testing schemes \((f_n, \varphi_n)\) with \(||f_n|| \ndot{\leq} e^{nR}\), and denote by \(E\) the achievable error exponent under the sequence.
By applying the schemes to the two hypotheses respectively, we obtain two joint distributions \(P_{MX^n}P_{Y^n}\) and \(Q_{MX^nY^n}\), where \(P_{Y^n} = P_{Y}^n\).
Following the multi-letter characterization proof \cite[Theorem 1]{ahlswedeHypothesisTestingCommunication1986}, we see that
\begin{equation}
    E \leq \liminf_{n\to \infty} \frac{1}{n}  D(P_{M}P_{Y^n} \| Q_{MY^n}).  \label{eq:converse_proof_thm_testing_against_dependence:multi-letter_characterization}
\end{equation}
Our goal is to derive a single-letter upper bound for \eqref{eq:converse_proof_thm_testing_against_dependence:multi-letter_characterization}.

To this end, consider a distribution \(\tilde{P}_{MX^nY^n}\) that satisfies the following conditions
\begin{align}
    \tilde{P}_{MY^n} & = P_{M}P_{Y^n}, \label{eq:converse_proof_thm_testing_against_dependence:two_conditions_MXY_1} \\
    \tilde{P}_{MX^n} & = P_{MX^n}. \label{eq:converse_proof_thm_testing_against_dependence:two_conditions_MXY_2}
\end{align}
The reasons for imposing the two conditions will become clear shortly.
Notice that
\begin{align}
    & D(\tilde{P}_{MX^nY^n} \| Q_{MX^nY^n}) \nonumber \\
    & = D(P_{M}P_{Y^n} \tilde{P}_{X^n | M Y^n} \| Q_{MY^n} Q_{X^n|MY^n}) \label{eq:converse_proof_thm_testing_against_dependence:use_conditoin_MXY_1} \\
    & = D(P_{M}P_{Y^n} \| Q_{MY^n}) + D( \tilde{P}_{X^n | M Y^n} \| Q_{X^n|MY^n} | P_{M}P_{Y^n}) \\
    & \geq D(P_{M}P_{Y^n} \| Q_{MY^n}),
\end{align}
where \eqref{eq:converse_proof_thm_testing_against_dependence:use_conditoin_MXY_1} is due to the condition in \eqref{eq:converse_proof_thm_testing_against_dependence:two_conditions_MXY_1}.
Hence, we get
\begin{align}
    & D(P_{M}P_{Y^n} \| Q_{MY^n}) \nonumber \\
    & \leq D(\tilde{P}_{MX^nY^n} \| Q_{MX^nY^n}) \\
    & = D(P_{MX^n} \tilde{P}_{Y^n | MX^n} \| Q_{MX^n} Q_{Y^n|X^n} ) \label{eq:converse_proof_thm_testing_against_dependence:use_conditoin_MXY_2} \\
    & = D(P_{MX^n} \| Q_{MX^n}) + D(\tilde{P}_{Y^n|MX^n} \| Q_{Y^n | X^n} | P_{MX^n})  \label{eq:converse_proof_thm_testing_against_dependence:splitting_M_X} \\
    & = D(P_{X^n} \| Q_{X^n}) + D(\tilde{P}_{Y^n|MX^n} \| Q_{Y^n | X^n} | P_{MX^n}), \label{eq:converse_proof_thm_testing_against_dependence:split_P_X_Q_X} 
\end{align}
where \eqref{eq:converse_proof_thm_testing_against_dependence:use_conditoin_MXY_2} is due to the condition in \eqref{eq:converse_proof_thm_testing_against_dependence:two_conditions_MXY_2}; \eqref{eq:converse_proof_thm_testing_against_dependence:split_P_X_Q_X} is because \(P_{M|X^n} = Q_{M|X^n} = \idc\{M = f_n(X^n)\}\).
Notice that \eqref{eq:converse_proof_thm_testing_against_dependence:split_P_X_Q_X}  holds for any \(\tilde{P}_{Y^n|MX^n}\), as long as the joint distribution \(P_{MX^n}\tilde{P}_{Y^n|MX^n}\) satisfies \eqref{eq:converse_proof_thm_testing_against_dependence:two_conditions_MXY_1}.
Thus, the task reduces to choosing a suitable \(\tilde{P}_{Y^n|MX^n}\) that can yield the desired single-letter upper bound.

We choose a suitable \(\tilde{P}_{Y^n|MX^n}\) as follows.
Given the distribution \(P_{MX^n}\), we can obtain its marginal distribution \(P_{MX^{i}}\) for every \(i\in [n]\).
Let \(U_i \triangleq (M, X^{i-1})\) and hence \(P_{MX^{i}} = P_{U_i X_i}\).
For every \(i \in [n]\), we then select a \(\tilde{P}_{Y_i | U_i X_i}\) such that the joint distribution \(\tilde{P}_{U_i X_i Y_i} = P_{U_i X_i}\tilde{P}_{Y_i | U_i X_i}\) satisfies \(\tilde{P}_{U_i Y_i} = P_{U_i} P_{Y_i}\), i.e., \(\tilde{P}_{Y_i} = P_{Y_i}\) and also \(U_i\) and \(Y_i\) are independent under \(\tilde{P}_{U_i Y_i}\).
Such \(\tilde{P}_{Y_i | U_i X_i}\) must exist since we can at least choose \(\tilde{P}_{Y_i | U_i X_i} = P_{Y_i}\).
Consider the \(\tilde{P}_{Y^n|MX^n}\) defined by
\begin{equation}
    \tilde{P}_{Y^n|MX^n} \triangleq \prod_{i=1}^{n} \tilde{P}_{Y_i | U_i X_i} = \prod_{i=1}^{n} \tilde{P}_{Y_i | M X^{i-1}X_i},  \label{eq:converse_proof_thm_testing_against_dependence:define_suitable_product_distribution}
\end{equation}
where \(\tilde{P}_{Y_i | U_i X_i}\) is the one we just selected.
In order to substitute this \(\tilde{P}_{Y^n|MX^n}\) into \eqref{eq:converse_proof_thm_testing_against_dependence:split_P_X_Q_X}, we need to verify that \(\tilde{P}_{MX^nY^n} = P_{MX^n}\tilde{P}_{Y^n|MX^n}\) satisfies \eqref{eq:converse_proof_thm_testing_against_dependence:two_conditions_MXY_1}.
In what follows, we write \(\tilde{P}_{MY^n}\) for \(\tilde{P}_{MY^n}(m,y^n)\), and adopt the same convention for all other distributions, e.g., \(P_{X^n|M}\) for \(P_{X^n|M}(x^n|m)\).
First, observe that
\begin{align}
    & \tilde{P}_{MY^n}  \nonumber \\
    & = \sum_{x^n} \tilde{P}_{MX^nY^n}\\
    & = P_{M}\sum_{x^n} P_{X^n | M } \tilde{P}_{Y^n | M X^n}  \\
    & = P_{M}\sum_{x^n} \prod_{i=1}^{n} P_{X_i | M X^{i-1} }  \prod_{i=1}^{n}\tilde{P}_{Y_i | M X^{i-1} X_i}  \\
    & = P_{M}\sum_{x^n} \prod_{i=1}^{n} P_{X_i | M X^{i-1} }  \tilde{P}_{Y_i | M X^{i-1} X_i} \\
    & = P_{M} \sum_{x^{n-1}} \prod_{i=1}^{n-1} P_{X_i | M X^{i-1} }  \tilde{P}_{Y_i | M X^{i-1} X_i} \times \big( \sum_{x_n} P_{X_n | M X^{n-1} }  \tilde{P}_{Y_n | M X^{n-1} X_n}  \big). \label{eq:converse_proof_thm_testing_against_dependence:sum_x_n}
\end{align}
For every fixed \((M, X^{n-1}) = (m, x^{n-1})\), we have
\begin{equation}
    \sum_{x_n} P_{X_n | M X^{n-1} }  \tilde{P}_{Y_n | M X^{n-1} X_n}  = \tilde{P}_{Y_n | M X^{n-1}}.
\end{equation}
Recall that \(U_n = (M, X^{n-1})\), and \(\tilde{P}_{Y_n | U_n X_n}\) is selected such that \(\tilde{P}_{U_nY_n} = P_{U_n}P_{Y_n} \).
We then see that
\begin{align}
    \sum_{x_n} P_{X_n | M X^{n-1} }  \tilde{P}_{Y_n | M X^{n-1} X_n}  & = \tilde{P}_{Y_n | M X^{n-1}} \\
    & = \tilde{P}_{Y_n | U_n} \\
    & = P_{Y_n}. \label{eq:converse_proof_thm_testing_against_dependence:independence_Y_U}
\end{align}
Substituting \eqref{eq:converse_proof_thm_testing_against_dependence:independence_Y_U} into \eqref{eq:converse_proof_thm_testing_against_dependence:sum_x_n} yields
\begin{align}
     & \tilde{P}_{MY^n} \nonumber \\
     & = P_{M} \times P_{Y_n} \times \sum_{x^{n-1}} \prod_{i=1}^{n-1} P_{X_i | M X^{i-1} } \tilde{P}_{Y_i | M X^{i-1} X_i} \\
     & = P_{M} \times P_{Y_n} \times \sum_{x^{n-2}} \prod_{i=1}^{n-2} P_{X_i | M X^{i-1} }  \tilde{P}_{Y_i | M X^{i-1} X_i} \times \big( \sum_{x_{n-1}} P_{X_{n-1} | M X^{n-2} }  \tilde{P}_{Y_{n-1} | M X^{n-2} X_{n-1}}   \big).
\end{align}
By repeating the backward summation, we get
\begin{equation}
    \tilde{P}_{MY^n} = P_{M} \times \prod_{i=1}^{n}P_{Y_i} = P_{M}P_{Y^n}.
\end{equation}
Therefore, the \(\tilde{P}_{Y_n | M X^n}\) defined in \eqref{eq:converse_proof_thm_testing_against_dependence:define_suitable_product_distribution} can be applied to the upper bound in \eqref{eq:converse_proof_thm_testing_against_dependence:split_P_X_Q_X}.

Before proceeding with the proof, we simplify the notation slightly for convenience.
We will write
\begin{equation}
    \sum_{x} g(Q_X) \triangleq \sum_{x \in \mathcal{X}} g(Q_X(x)),
\end{equation}
for any alphabet \(\mathcal{X}\), pmf \(Q_X\), and function \(g\).
Similar notation will also apply to, e.g.,
\begin{equation}
    \sum_{x , y} g(Q_X, Q_Y) \triangleq \sum_{x \in \mathcal{X} , y \in \mathcal{Y}} g(Q_X(x), Q_Y(y)).
\end{equation}
With the above in mind, we now apply the \(\tilde{P}_{Y_n | M X^n}\) defined in \eqref{eq:converse_proof_thm_testing_against_dependence:define_suitable_product_distribution}  to \eqref{eq:converse_proof_thm_testing_against_dependence:split_P_X_Q_X} and obtain that
\begin{align}
    & D(P_{M}P_{Y^n} \| Q_{MY^n}) \nonumber \\
    & \leq D(P_{X^n} \| Q_{X^n}) + D(\tilde{P}_{Y^n|MX^n} \| Q_{Y^n | X^n} | P_{MX^n}) \\
    & = D(P_{X^n} \| Q_{X^n}) + \sum_{m, x^n, y^n} \tilde{P}_{MX^nY^n} \log \frac{ \tilde{P}_{Y^n|MX^n}  }{ Q_{Y^n|X^n}} \\
    & = D(P_{X^n} \| Q_{X^n}) + \sum_{m, x^n, y^n} \tilde{P}_{MX^nY^n} \log \frac{ \prod_{i=1}^{n} \tilde{P}_{Y_i | U_i X_i} }{ \prod_{i=1}^{n} Q_{Y_i|X_i}} \\
    & = D(P_{X^n} \| Q_{X^n})  + \sum_{i=1}^{n} \sum_{u_i, x_i, y_i} \tilde{P}_{U_i X_i Y_i} \log \frac{ \tilde{P}_{Y_i | U_i X_i}  }{ Q_{Y_i|X_i} } \\
    & = D(P_{X^n} \| Q_{X^n})  + \sum_{i=1}^{n} D(\tilde{P}_{Y_i | U_i X_i} \| Q_{Y_i|X_i} | P_{U_i X_i}) \\
    & = D(P_{X^n} \| Q_{X^n}) + n D(\tilde{P}_{Y_J | U_J X_J J} \| Q_{Y_J|X_J J} | P_{U_J X_J J}) \label{eq:converse_proof_thm_testing_against_dependence:define_J}.\\
    & = D(P_{X^n} \| Q_{X^n}) + nD(\tilde{P}_{JU_JX_JY_J} \| P_{J U_J X_J} Q_{Y_J|X_J J}) \\
    & = D(P_{X^n} \| Q_{X^n}) + nD(\tilde{P}_{JU_JX_JY_J} \| P_{J U_J | X_J} P_{X_J} Q_{Y_J|X_J J}), \label{eq:converse_proof_thm_testing_against_dependence:single_letterize_J}
\end{align}
where in \eqref{eq:converse_proof_thm_testing_against_dependence:define_J}, \(J\)  is uniformly distributed over \([n]\), known as a time-sharing random variable.
Let \(\bar{U} \triangleq (J, U_J)\), and notice that \(P_{X_J} = P_{X}\), \(P_{Y_J} = P_{Y}\), and \(Q_{Y_J | X_J, J} = Q_{Y|X}\) due to the i.i.d. assumption.
Thus, we can write  the upper bound in \eqref{eq:converse_proof_thm_testing_against_dependence:single_letterize_J} as
\begin{align}
    & \frac{1}{n}D(P_{M}P_{Y^n} \| Q_{MY^n}) \nonumber \\
    & \leq \frac{1}{n}D(P_{X^n} \| Q_{X^n}) + D(\tilde{P}_{\bar{U}XY} \| P_{\bar{U} | X} P_{X} Q_{Y|X}) \\
    & = D(P_{X} \| Q_{X}) + D(\tilde{P}_{\bar{U}XY} \| P_{\bar{U}| X} P_{X} Q_{Y|X}) \\
    & = D(\tilde{P}_{\bar{U}XY} \| P_{\bar{U} | X} Q_{XY}). \label{eq:converse_proof_thm_testing_against_dependence:bound_in_U}
\end{align}
Since \(\tilde{P}_{U_i X_i Y_i} = P_{U_iX_i} \tilde{P}_{Y_i | U_i X_i}\) satisfies \(\tilde{P}_{U_iY_i} = P_{U_i}P_{Y_i}\) for every \(i\), we see that \(\tilde{P}_{\bar{U}XY}  \) satisfies
\begin{align}
    \tilde{P}_{\bar{U}X} & = P_{\bar{U}X}, \label{eq:converse_proof_thm_testing_against_dependence:marginal_condition_1} \\
    \tilde{P}_{\bar{U}Y} & = P_{\bar{U}}P_{Y}.\label{eq:converse_proof_thm_testing_against_dependence:marginal_condition_2}
\end{align}
We have freedom in selecting \(\tilde{P}_{Y_i | U_i X_i}\) and hence \(\tilde{P}_{U_iX_iY_i}\), so we can minimize over it for every \(i\) to obtain the tightest upper bound, yielding
\begin{align}
    \frac{1}{n}D(P_{M}P_{Y^n} \| Q_{MY^n}) & \leq \min_{\tilde{P}_{\bar{U}XY}} D(\tilde{P}_{\bar{U}XY} \| P_{\bar{U}|X}Q_{XY}) \label{eq:converse_proof_thm_testing_against_dependence:inner_minimization} \\
    & = E(P_{\bar{U}X}P_{Y} \| Q_{\bar{U} XY}), \label{eq:converse_proof_thm_testing_against_dependence:E_notation}
\end{align}
where in \eqref{eq:converse_proof_thm_testing_against_dependence:inner_minimization}, \(\tilde{P}_{\bar{U}XY}\) needs to satisfy the conditions in \eqref{eq:converse_proof_thm_testing_against_dependence:marginal_condition_1} and \eqref{eq:converse_proof_thm_testing_against_dependence:marginal_condition_2}.

By the familiar weak converse technique, for any \(\tau > 0\) and sufficiently large \(n\), we have
\begin{align}
    R + \tau & \geq \frac{1}{n} H(M) \\
    & \geq \frac{1}{n}I(M; X^n) \\
    & = \frac{1}{n}\sum_{i=1}^{n}I(M;X_i|X^{i-1}) \\
    & = \frac{1}{n} \sum_{i=1}^{n}I(M, X^{i-1};X_i) \\
    & = \frac{1}{n} \sum_{i=1}^{n}I(U_i;X_i)  \\
    & = I( U_J ; X_J | J) \\
    & = I(U_J, J; X_J) \\
    & = I(\bar{U}; X).
\end{align}
Therefore, we can upper bound \eqref{eq:converse_proof_thm_testing_against_dependence:E_notation} by maximizing over all \(U\) satisfying \(I(U;X) \leq R + \tau\) and conclude that
\begin{align}
    \frac{1}{n}D(P_{M}P_{Y^n} \| Q_{MY^n}) & \leq \max_{ \substack{P_{U|X}: \\ I(U;X) \leq R + \tau} } E(P_{UX}P_{Y} \| Q_{UXY}). \label{eq:converse_proof_thm_testing_against_dependence:tau}
\end{align}
Notice that \eqref{eq:converse_proof_thm_testing_against_dependence:tau} holds for any \(\tau > 0\) as \(n \to \infty\).
The cardinality bound \(|\mathcal{U}| \leq |\mathcal{X}| + 1\) is due to Han \cite[Theorem 3]{hanHypothesisTestingMultiterminal1987}.
With this, the converse proof is completed.
Note that in the rest of this paper, whenever we apply a weak converse argument, we omit \(\tau>0\) and simply write \(R\ge \frac1n H(M)\), with the understanding that the same argument applies.

\section{Proof of Theorem \ref{thm:testing_against_product_dependence_independence}}
\label{sec:proof_testing_against_product_dependence_independence}

\subsection{Achievability and Simplified Expression}
Denote by \(E_{\mathrm{p}}(R)\) the RHS of \eqref{eq:main_results:testing_against_product_dependence_independence:exponent}.
Since we know \(E(R) \geq E_{\mathrm{H}}(R)\), our goal is to show \(E_{\mathrm{H}}(R) \geq E_{\mathrm{p}}(R)\).
Recall the definition of \(E_{\mathrm{H}}(R)\) in \eqref{eq:problem_setup:Han_exponent_definition_1} and \eqref{eq:problem_setup:Han_exponent_definition_2}.
Under the two hypotheses \(\mathcal{H}_0: P_{X_1}P_{Y_1}P_{X_2Y_2}\) and \(\mathcal{H}_1: Q_{X_1Y_1} Q_{X_2}Q_{Y_2}\), we have
\begin{align}
    E_{\mathrm{H}}(R) & = \max_{ \substack{P_{U_1U_2|X_1,X_2}: \\ I(U_1, U_2;X_1, X_2) \leq R }} E(P_{U_1U_2|X_1X_2}P_{X_1}P_{Y_1}P_{X_2Y_2} \|P_{U_1U_2|X_1X_2} Q_{X_1Y_1} Q_{X_2}Q_{Y_2}) \label{eq:proof_testing_against_product_dependence_independence:achievability:Han_exponent_U_1_U_2} \\
    & \geq \max_{ \substack{P_{U_1|X_1}, P_{U_2|X_2}: \\ I(U_1;X_1) + I(U_2;X_2) \leq R }} E(P_{U_1X_1}P_{Y_1}P_{U_2X_2Y_2} \|Q_{U_1X_1Y_1} Q_{U_2X_2}Q_{Y_2}) ,\label{eq:proof_testing_against_product_dependence_independence:achievability:product_U_1_U_2}
\end{align}
where in \eqref{eq:proof_testing_against_product_dependence_independence:achievability:Han_exponent_U_1_U_2}, we let \(U = (U_1, U_2)\) without loss of generality; \eqref{eq:proof_testing_against_product_dependence_independence:achievability:product_U_1_U_2} is because we now only consider the product case \(P_{U_1U_2|X_1X_2} = P_{U_1|X_1}P_{U_2|X_2}\).
To proceed, let us recall that
\begin{align}
    D(P_{XY} \| Q_{X}Q_{Y}) & = D(P_{X} \| Q_{X}) + D(P_{Y|X}\|Q_{Y}|P_X) \\
    & \geq D(P_{X} \| Q_{X}) + D(P_{Y}\|Q_{Y}), \label{eq:eq:proof_testing_against_product_dependence_independence:achievability:divergence_rule}
\end{align}
where \eqref{eq:eq:proof_testing_against_product_dependence_independence:achievability:divergence_rule} is due to the data processing inequality.
Then, we see that
\begin{align}
    & E(P_{U_1X_1}P_{Y_1}P_{U_2X_2Y_2} \|Q_{U_1X_1Y_1} Q_{U_2X_2}Q_{Y_2}) \nonumber \\
    & = \min_{ \substack{ \tilde{P}_{U_1U_2X_1X_2Y_1Y_2} \in \\ \mathcal{P}_{\mathrm{H}}(P_{U_1X_1}P_{Y_1}P_{U_2X_2Y_2}) }} D(\tilde{P}_{U_1U_2X_1X_2Y_1Y_2} \| Q_{U_1X_1Y_1}Q_{U_2X_2}Q_{Y_2} ) \\
    & \geq \min_{ \substack{ \tilde{P}_{U_1U_2X_1X_2Y_1Y_2} \in \\ \mathcal{P}_{\mathrm{H}}(P_{U_1X_1}P_{Y_1}P_{U_2X_2Y_2}) }} D(\tilde{P}_{U_1X_1Y_1} \| Q_{U_1X_1Y_1} ) + D(\tilde{P}_{U_2X_2Y_2} \| Q_{U_2X_2}Q_{Y_2}) \label{eq:eq:proof_testing_against_product_dependence_independence:achievability:applying_divergence_rule}\\
    & \geq \min_{ \substack{ \tilde{P}_{U_1X_1Y_1} \in \\ \mathcal{P}_{\mathrm{H}}(P_{U_1X_1}P_{Y_1}) }} D(\tilde{P}_{U_1X_1Y_1} \| Q_{U_1X_1Y_1} ) + \min_{ \substack{ \tilde{P}_{U_2X_2Y_2} \in \\ \mathcal{P}_{\mathrm{H}}(P_{U_2X_2Y_2}) }} D(\tilde{P}_{U_2X_2Y_2} \|Q_{U_2X_2}Q_{Y_2}) \label{eq:eq:proof_testing_against_product_dependence_independence:achievability:definition_calP_H} \\
    & = E(P_{U_1X_1}P_{Y_1} \| Q_{U_1X_1Y_1}) +  E(P_{U_2X_2Y_2} \| Q_{U_2X_2}Q_{Y_2}), \label{eq:eq:proof_testing_against_product_dependence_independence:achievability:definition_E_divergence} 
\end{align}
where \eqref{eq:eq:proof_testing_against_product_dependence_independence:achievability:applying_divergence_rule} is due to \eqref{eq:eq:proof_testing_against_product_dependence_independence:achievability:divergence_rule}; \eqref{eq:eq:proof_testing_against_product_dependence_independence:achievability:definition_calP_H} is because \(\min_{x} [f(x) + g(x)] \geq \min_{x}f(x) + \min_{x} g(x)\) and the definition of \(\mathcal{P}_{\mathrm{H}}\) in \eqref{eq:problem_setup:Han_exponent_definition_2}.
We obtain \(E_{\mathrm{H}}(R) \geq E_{\mathrm{p}}(R)\) after substituting \eqref{eq:eq:proof_testing_against_product_dependence_independence:achievability:definition_E_divergence} into \eqref{eq:proof_testing_against_product_dependence_independence:achievability:product_U_1_U_2}.

\begin{remark}
    Following the approach in \cite[Section VI]{watanabeSuboptimalityRandomBinning2022}, we can establish the achievability of \(E_{\mathrm{p}}(R)\)  more directly as follows: we separate the two components \((X_1^n, Y_1^n)\) and \((X_2^n, Y_2^n)\), and apply the quantization scheme to each component using \(P_{U_1|X_1}\) and \(P_{U_2|X_2}\) independently; we accept \(\mathcal{H}_0\) only if both schemes accept \(\mathcal{H}_0\).
    Since the decisions of the two schemes are independent, the type-II error probability of the above scheme is given by the product of the type-II error probabilities of the two individual schemes, yielding \(E_{\mathrm{p}}(R)\).
\end{remark}

\subsection{Converse}
We next establish \(E(R) \leq E_{\mathrm{p}}(R)\).
Consider a sequence of testing schemes \((f_n, \varphi_n)\) with \(||f_n|| \dot{\leq} e^{nR}\), and denote by \(E\) the achievable error exponent under the sequence.
By applying the schemes to the two hypotheses  \(\mathcal{H}_0: P_{X_1}P_{Y_1}P_{X_2Y_2}\) and \(\mathcal{H}_1: Q_{X_1Y_1} Q_{X_2}Q_{Y_2}\) respectively, we obtain two joint distributions \(P_{M}P_{Y_1^n}P_{Y_2^n|M}\) and \(Q_{M}Q_{Y_1^n|M}Q_{Y_2^n}\).
Following the multi-letter characterization proof \cite[Theorem 1]{ahlswedeHypothesisTestingCommunication1986}, we see that
\begin{align}
    E &\leq \liminf_{n\to \infty} \frac{1}{n} D( P_{M}P_{Y_1^n}P_{Y_2^n|M} \| Q_{M}Q_{Y_1^n|M}Q_{Y_2^n}).
\end{align}
First, notice that
\begin{equation}
    \frac{1}{n} D( P_{M}P_{Y_1^n}P_{Y_2^n|M} \| Q_{M}Q_{Y_1^n|M}Q_{Y_2^n}) = \frac{1}{n} D( P_{M}P_{Y_1^n} \| Q_{MY_1^n}) + \frac{1}{n} D( P_{MY_2^n} \| P_{M}Q_{Y_2^n}). \label{eq:proof_testing_against_product_dependence_independence:converse:divergence}
\end{equation}
Using the same arguments as in Section \ref{sec:converse_proof_thm_testing_against_dependence}, and in particular \eqref{eq:converse_proof_thm_testing_against_dependence:splitting_M_X}, we get
\begin{align}
    & D( P_{M}P_{Y_1^n} \| Q_{MY_1^n}) \nonumber \\
    & \leq D(P_{MX_1^n} \| Q_{MX_1^n}) + D(\tilde{P}_{Y_1^n|MX_1^n} \| Q_{Y_1^n | X_1^n} | P_{MX_1^n}) \\
    & \leq D(P_{MX_1^nX_2^n} \| Q_{MX_1^nX_2^n}) + D(\tilde{P}_{Y_1^n|MX_1^n} \| Q_{Y_1^n | X_1^n} | P_{MX_1^n}) \label{eq:proof_testing_against_product_dependence_independence:converse:data_processing}\\
    & = D(P_{X_1^nX_2^n} \| Q_{X_1^nX_2^n}) + D(\tilde{P}_{Y_1^n|MX_1^n} \| Q_{Y_1^n | X_1^n} | P_{MX_1^n}) \label{eq:proof_testing_against_product_dependence_independence:converse:encoding_mapping}\\
    & = D(P_{X_1^n} \| Q_{X_1^n}) + D(P_{X_2^n} \| Q_{X_2^n}) + D(\tilde{P}_{Y_1^n|MX_1^n} \| Q_{Y_1^n | X_1^n} | P_{MX_1^n}), \label{eq:proof_testing_against_product_dependence_independence:converse:M_X_1_X_2}
\end{align}
where \eqref{eq:proof_testing_against_product_dependence_independence:converse:data_processing} is due to the data processing inequality; \eqref{eq:proof_testing_against_product_dependence_independence:converse:encoding_mapping} is because \(P_{M|X_1^n X_2^n} = Q_{M|X_1^n X_2^n} = \idc\{M = \tilde{f}_n(X_1^n, X_2^n)\}\).
Substituting \eqref{eq:proof_testing_against_product_dependence_independence:converse:M_X_1_X_2} into \eqref{eq:proof_testing_against_product_dependence_independence:converse:divergence} yields that
\begin{align}
    & \frac{1}{n} D( P_{M}P_{Y_1^n} \| Q_{MY_1^n}) + \frac{1}{n} D( P_{MY_2^n} \| P_{M}Q_{Y_2^n}).  \nonumber \\
    & \leq  \frac{1}{n}D(P_{X_1^n} \| Q_{X_1^n}) + \frac{1}{n}D(P_{X_2^n} \| Q_{X_2^n}) + \frac{1}{n}D(\tilde{P}_{Y_1^n|MX_1^n} \| Q_{Y_1^n | X_1^n} | P_{MX_1^n}) + \frac{1}{n} D( P_{MY_2^n} \| P_{M}Q_{Y_2^n}) \\
    & =  \frac{1}{n}D(P_{X_1^n} \| Q_{X_1^n}) + \frac{1}{n}D(\tilde{P}_{Y_1^n|MX_1^n} \| Q_{Y_1^n | X_1^n} | P_{MX_1^n})  \nonumber \\
     & \hspace{3cm} + \frac{1}{n}D(P_{X_2^n} \| Q_{X_2^n}) + \frac{1}{n} D(P_{Y_2^n} \| Q_{Y_2^n}) + \frac{1}{n} I(M;Y_2^n). \label{eq:proof_testing_against_product_dependence_independence:converse:sum_two_parts}
\end{align}
On the other hand, we have
\begin{align}
    R & \geq \frac{1}{n}H (M) \\
    & \geq \frac{1}{n}I(M; X_1^n, X_2^n) \\
    & = \frac{1}{n}I(M;X_1^n) + \frac{1}{n} I(M; X_2^n|X_1^n) \\
    & \geq \frac{1}{n}I(M;X_1^n) + \frac{1}{n} I(M; X_2^n), \label{eq:proof_testing_against_product_dependence_independence:converse:rate}
\end{align}
where \eqref{eq:proof_testing_against_product_dependence_independence:converse:rate} is because \(X_1^n\) and \(X_2^n\) are independent.
From \eqref{eq:proof_testing_against_product_dependence_independence:converse:sum_two_parts} and \eqref{eq:proof_testing_against_product_dependence_independence:converse:rate}, we now can establish \(E(R) \leq E_{\mathrm{p}}(R)\) by applying the reasoning in Section \ref{sec:converse_proof_thm_testing_against_dependence} and Appendix \ref{apd:converse_proof_generalization_testing_against_independence} to \((X_1^n, Y_1^n)\) and \((X_2^n, Y_2^n)\) separately, in particular the arguments following \eqref{eq:converse_proof_thm_testing_against_dependence:split_P_X_Q_X} and \eqref{eq:apd:converse_proof_generalization_testing_against_independence:M_Y^n}.
The cardinality bounds \(|\mathcal{U}_1| \leq |\mathcal{X}_1| + 1\) and \(|\mathcal{U}_2| \leq |\mathcal{X}_2| + 1\) also follow from applying Han's proof of \cite[Theorem 3]{hanHypothesisTestingMultiterminal1987} to \(U_1\) and \(U_2\) separately.
With this, the converse proof is completed.

\section{Achievability Proof of Theorem \ref{thm:conditional_testing_against_dependence}}
\label{sec:achievability_proof_conditional_testing_against_depedence}

Consider \(P_{XYZ} = P_{XZ}P_{Y|Z}\), and recall that \(P_{UXYZ} = P_{U|XZ}P_{XYZ}\) and \(Q_{UXYZ} = P_{U|XZ} Q_{XYZ}\).
In this section, we show that for conditional testing against dependence, we have
\begin{equation}
    \tilde{E}(R) \geq \max_{\substack{P_{U|XZ}: \\ I(U;X|Z) \leq R }} \tilde{E}(P_{UXYZ} \| Q_{UXYZ}).
\end{equation}
Before proceeding to the proof, we introduce some notation and discuss a preliminary result that will be useful down the line.

\subsection{Notation and Preliminaries}
Given a pmf \(P_{X}\) and small constant \(\epsilon > 0\), we write \(Q_X \overset{\epsilon}{\sim} P_{X}\) if
\begin{equation}
    |Q_X(a) - P_{X}(a)| \leq \epsilon P_{X}(a), \qquad \forall a \in \mathcal{X}.
\end{equation}
We define the \(\epsilon\)-typical set \(\mathcal{T}_n^{\epsilon}(P_X)\) as the set of all \(\bm{x} \in \mathcal{X}^n\) satisfying \(\hat{P}_{\bm{x}} \overset{\epsilon}{\sim} P_{X}\).
For a joint pmf \(P_{XY}\) and an \(\bm{x} \in \mathcal{T}_n^{\epsilon}(P_X)\), we define the conditionally  \(\epsilon\)-typical set \(\mathcal{T}_n^{\epsilon}(P_{Y|X} | \bm{x})\) as
\begin{equation}
    \mathcal{T}_n^{\epsilon}(P_{Y|X} | \bm{x}) \triangleq \{ \bm{y} \in \mathcal{Y}^n: (\bm{x}, \bm{y}) \in \mathcal{T}_n^{\epsilon}(P_{XY}) \}.
\end{equation}
The typicality notion adopted here is known as robust typicality; see \cite[Chapter 2]{gamalNetworkInformationTheory2011} for its properties.

We next review a key tool, known as the type covering lemma, that will be used in the design of the coding scheme further on.
The type covering lemma was proposed by Berger \cite{bergerRateDistortionTheory1971} to show the existence of good codes for lossy source coding.
To account for the additional side information sequence \(Z^n\) in our setting, we need the following modified version.
Note that a similar version has appeared in the literature, see \cite[Lemma 3.8]{graczykGuessWhat2021}.

\begin{lemma}[Conditional Type Covering]
\label{lem:condition_type_covering}
For every pmf \(P_{UXZ}\) and sequence \(\bm{z} \in \mathcal{T}_n^{\epsilon}(P_{Z})\), there exists a set \(\mathcal{A}_n \subseteq \mathcal{T}_n^{\epsilon}(P_{U|Z} | \bm{z}) \) with
\begin{equation}
    |\mathcal{A}_n| \ndot{\leq} e^{nI(U;X|Z)}
\end{equation}
such that for every \(\bm{x} \in \mathcal{T}_n^{\epsilon}(P_{X|Z} | \bm{z})\) we can find \(\bm{u} \in \mathcal{A}_n\) satisfying \((\bm{u}, \bm{x}, \bm{z}) \in \mathcal{T}_n^{\epsilon}(P_{UXZ})\).
\end{lemma}
\begin{proof}
    Under the robust typicality, if \((\bm{u}, \bm{x}, \bm{z}) \in \mathcal{T}_n^{\epsilon}(P_{UXZ})\), then we must have \(\bm{u} \in \mathcal{T}_n^{\epsilon}(P_{U|Z} | \bm{z})\).
    Hence, Lemma \ref{lem:condition_type_covering} can be proved using random coding over the set \(\mathcal{T}_n^{\epsilon}(P_{U|Z} | \bm{z})\), following a similar argument as in the proof of \cite[Lemma 3.8]{graczykGuessWhat2021}.
\end{proof}
With the above prelude, we now move on to the main part of the proof.

\subsection{Coding Scheme and Error Analysis}
We first present the coding scheme.
Recall the two hypotheses \(\mathcal{H}_0: P_{XYZ}\) and \(\mathcal{H}_1: Q_{XYZ}\).
We begin by selecting a conditional distribution \(P_{U|XZ}\) such that \(I(U;X|Z) \leq R\) under \(P_{UXYZ} = P_{U|XZ}P_{XYZ}\).
For every \(\bm{z} \in \mathcal{T}_n^{\epsilon}(P_{Z})\), we select a codebook \(\mathcal{C}_n(\bm{z}) = \{\bm{u}_1, \bm{u}_2, \ldots \} \subseteq \mathcal{T}_n^{\epsilon}(P_{U|Z}|\bm{z})\) according to Lemma \ref{lem:condition_type_covering}.
The collection of the codebooks \(\{\mathcal{C}_n(\bm{z})\}_{\bm{z} \in \mathcal{T}_n^{\epsilon}(P_{Z}) }\) is available to both the transmitter and the receiver.
The coding procedure now proceeds as follows.

The transmitter first observes \((\bm{x}, \bm{z})\) and examines whether \((\bm{x}, \bm{z}) \in \mathcal{T}_n^{\epsilon}(P_{XZ})\) or not. We will count the event \((\bm{x}, \bm{z}) \notin \mathcal{T}_n^{\epsilon}(P_{XZ})\) as a type-II error event in our analysis, i.e., the receiver will decide \(\mathcal{H}_1\) is true in this case.
If \((\bm{x}, \bm{z}) \in \mathcal{T}_n^{\epsilon}(P_{XZ})\), then the transmitter identifies a codeword \(\bm{u}_{m} \in \mathcal{C}_n(\bm{z})\) such that \((\bm{u}_{m}, \bm{x}, \bm{z}) \in \mathcal{T}_n^{\epsilon}(P_{UXZ})\).
If there are multiple such codewords, the transmitter selects one of them arbitrarily.
The transmitter sends \(M = m\) to the receiver.
Since \(I(U;X|Z) \leq R\), we see \(|\mathcal{C}_n(\bm{z})| \ndot{\leq} e^{nR}\) and hence the rate limit \(R\) is satisfied asymptotically.
Given the index \(m\) and observation \(\bm{z}\), the receiver can determine \(\bm{u}_{m} \in \mathcal{C}_n(\bm{z})\).
The receiver decides \(\mathcal{H}_0\) is true if \((\bm{u}_m, \bm{y}, \bm{z}) \in \mathcal{T}_n^{\epsilon^{\prime}}(P_{UYZ})\)  for an \(\epsilon^{\prime} > \epsilon\).

By the conditional typicality lemma in \cite[Chapter 2]{gamalNetworkInformationTheory2011}, we can verify that the type-I error probability \(\alpha_n\)  of the above scheme goes to \(0\) as \(n\) grows.
We henceforth focus on analyzing the type-II error probability, i.e., the probability that the receiver declares \(\mathcal{H}_0\) when \(\mathcal{H}_1\)  is true.
For simplicity, we will write \(\bm{U} \triangleq \bm{u}_{M}\) in the following, which denotes the quantization codeword.
Suppose \(\mathcal{H}_1: Q_{XYZ}\) is true. 
It follows that
\begin{align}
    \beta_n & = \Pr \{ \varphi_n( \bm{Y}, \bm{Z}, M ) = \mathcal{H}_0  \} \\
    &  =  \Pr \{ (\bm{X}, \bm{Z}) \in \mathcal{T}_n^{\epsilon}(P_{XZ}),  (\bm{U}, \bm{Y}, \bm{Z}) \in \mathcal{T}_n^{\epsilon^{\prime}}(P_{UYZ}) \} \\
    & =   \Pr \{ (\bm{X}, \bm{Z}) \in \mathcal{T}_n^{\epsilon}(P_{XZ}) \}  \Pr  \{ (\bm{U}, \bm{Y}, \bm{Z}) \in \mathcal{T}_n^{\epsilon^{\prime}}(P_{UYZ}) | (\bm{X}, \bm{Z}) \in \mathcal{T}_n^{\epsilon}(P_{XZ}) \}. \label{eq:conditiona_testing_against_dependence:achievability:conditional_probability_expansion}
\end{align}
It is clear that
\begin{equation}
    \Pr \{ (\bm{X}, \bm{Z}) \in \mathcal{T}_n^{\epsilon}(P_{XZ}) \} \ndot{=} e^{-nD(P_{XZ} \| Q_{XZ})}. \label{eq:conditiona_testing_against_dependence:achievability:conditional_probability_bound_1}
\end{equation}
On the other hand, since we know \((\bm{U}, \bm{X}, \bm{Z}) \in \mathcal{T}_n^{\epsilon}(P_{UXZ})\), we see that
\begin{align}
    & \Pr  \{ (\bm{U}, \bm{Y}, \bm{Z}) \in \mathcal{T}_n^{\epsilon^{\prime}}(P_{UYZ}) | (\bm{X}, \bm{Z}) \in \mathcal{T}_n^{\epsilon}(P_{XZ}) \} \nonumber \\
    & = \Pr \{ \bm{Y}\in \mathcal{T}_n^{\epsilon^{\prime}}(P_{Y|UZ}| \bm{U}, \bm{Z}) | (\bm{U}, \bm{X}, \bm{Z}) \in \mathcal{T}_n^{\epsilon}(P_{UXZ}) \}
\end{align}
Notice that we can write
\begin{equation}
    \mathcal{T}_n^{\epsilon^{\prime}}(P_{Y|UZ}| \bm{U}, \bm{Z}) = \bigcup_{\tilde{P}_{Y|UXZ}} \mathcal{T}_n(\tilde{P}_{Y|UXZ}| \bm{U}, \bm{X}, \bm{Z}),
\end{equation}
where the union is taken over all conditional type \(\tilde{P}_{Y|UXZ}\) such that \(\tilde{P}_{UXYZ} = \tilde{P}_{Y|UXZ} \hat{P}_{\bm{U} \bm{X} \bm{Z}}\) satisfies  \(\tilde{P}_{Y|UZ}\hat{P}_{\bm{U}\bm{Z}} \overset{\epsilon^{\prime}}{\sim} P_{Y|UZ}\hat{P}_{\bm{U}\bm{Z}}\).
Thus, under \(\mathcal{H}_1: Q_{XYZ}\), we get
\begin{align}
    &  \Pr \{ \bm{Y}\in \mathcal{T}_n^{\epsilon^{\prime}}(P_{Y|UZ}| \bm{U}, \bm{Z}) | (\bm{U}, \bm{X}, \bm{Z}) \in \mathcal{T}_n^{\epsilon}(P_{UXZ}) \} \nonumber \\
    & \leq \sum_{\tilde{P}_{Y|UXZ}} Q_{Y|XZ}^n[ \bm{Y}\in \mathcal{T}_n^{\epsilon^{\prime}}(\tilde{P}_{Y|UXZ}| \bm{U}, \bm{X}, \bm{Z}) | (\bm{U}, \bm{X}, \bm{Z}) \in \mathcal{T}_n^{\epsilon}(P_{UXZ})  ] \\
    & \ndot{=} \sum_{\tilde{P}_{Y|UXZ}} e^{-nD(\tilde{P}_{Y|UXZ} \| Q_{Y|XZ} | P_{UXZ})} \\
    & = \sum_{ \substack{\tilde{P}_{UXYZ} \in \tilde{\mathcal{P}}_{\mathrm{H}}(P_{UXYZ})  }} e^{-nD(\tilde{P}_{UXYZ} \| P_{UXZ}Q_{Y|XZ})}. \label{eq:conditiona_testing_against_dependence:achievability:conditional_probability_bound_2}
\end{align}
Note that
\begin{equation}
    D(P_{XZ} \| Q_{XZ}) + D(\tilde{P}_{UXYZ} \| P_{UXZ}Q_{Y|XZ}) = D(\tilde{P}_{UXYZ} \| P_{U|XZ}Q_{XYZ}).
\end{equation}
Hence, substituting \eqref{eq:conditiona_testing_against_dependence:achievability:conditional_probability_bound_1} and \eqref{eq:conditiona_testing_against_dependence:achievability:conditional_probability_bound_2} into \eqref{eq:conditiona_testing_against_dependence:achievability:conditional_probability_expansion} yields
\begin{align}
     \beta_n & \ndot{\leq}  \sum_{ \substack{\tilde{P}_{UXYZ} \in \tilde{\mathcal{P}}_{\mathrm{H}}(P_{UXYZ})  }} e^{-nD(\tilde{P}_{UXYZ} \| P_{U|XZ}Q_{XYZ})} \\
     & \ndot{=} \max_{ \substack{\tilde{P}_{UXYZ} \in \tilde{\mathcal{P}}_{\mathrm{H}}(P_{UXYZ})  }} e^{-nD(\tilde{P}_{UXYZ} \| P_{U|XZ}Q_{XYZ})} \\
     & = e^{-n\tilde{E}(P_{UXYZ} \| Q_{UXYZ})} \\
     & = \min_{\substack{P_{U|XZ}: \\ I(U;X|Z) \leq R}} e^{-n\tilde{E}(P_{UXYZ} \| Q_{UXYZ})}, \label{eq:conditiona_testing_against_dependence:achievability:optimization_over_P_UXZ}
\end{align}
where in \eqref{eq:conditiona_testing_against_dependence:achievability:optimization_over_P_UXZ}, we assume that in the coding scheme \(P_{U|XZ}\) is chosen to minimize the resulting upper bound.
This completes the achievability proof of Theorem \ref{thm:conditional_testing_against_dependence}.

\section{Converse Proof of Theorem \ref{thm:conditional_testing_against_dependence}}
\label{sec:converse_proof_conditional_testing_against_depedence}

Now that the achievability of Theorem \ref{thm:conditional_testing_against_dependence} is proved, we turn to the converse in this section.
The converse proof will closely follow that in Section \ref{sec:converse_proof_thm_testing_against_dependence}, with modifications to incorporate the additional side information sequence \(Z^n\).
We begin with the multi-letter characterization of the error exponent.
Consider a sequence of testing schemes \((\tilde{f}_n, \varphi_n)\) with \(||\tilde{f}_n|| \ndot{\leq} e^{nR}\), and denote by \(E\) the achievable error exponent under the sequence.
By applying the schemes to the two hypotheses respectively, we obtain two joint distributions \(P_{MX^nZ^n}P_{Y^n|Z^n}\) and \(Q_{MX^nY^nZ^n}\), where \(P_{Y^n|Z^n} = P_{Y|Z}^n\).
Following the multi-letter characterization proof \cite[Theorem 1]{ahlswedeHypothesisTestingCommunication1986}, we see that
\begin{equation}
    E \leq \liminf_{n\to \infty} \frac{1}{n} D(P_{MZ^n}P_{Y^n|Z^n} \| Q_{MY^nZ^n}).  \label{eq:converse_proof_conditional_testing_against_dependence:multi-letter_characterization}
\end{equation}
The goal is to derive a single-letter upper bound for \eqref{eq:converse_proof_conditional_testing_against_dependence:multi-letter_characterization}.

To this end, consider a distribution \(\tilde{P}_{MX^nY^nZ^n}\) that satisfies the following conditions
\begin{align}
    \tilde{P}_{MY^nZ^n} & = P_{MZ^n}P_{Y^n|Z^n}, \label{eq:converse_proof_conditional_testing_against_dependence:two_conditions_MXY_1} \\
    \tilde{P}_{MX^nZ^n} & = P_{MX^nZ^n}. \label{eq:converse_proof_conditional_testing_against_dependence:two_conditions_MXY_2}
\end{align}
Notice that
\begin{align}
    & D(\tilde{P}_{MX^nY^nZ^n} \| Q_{MX^nY^nZ^n}) \nonumber \\
    & = D(P_{MZ^n}P_{Y^n|Z^n} \tilde{P}_{X^n | M Y^n Z^n} \| Q_{MY^nZ^n} Q_{X^n|MY^nZ^n}) \label{eq:converse_proof_conditional_testing_against_dependence:use_conditoin_MXY_1} \\
    & = D(P_{MZ^n}P_{Y^n|Z^n} \| Q_{MY^nZ^n}) + D( \tilde{P}_{X^n | M Y^n Z^n} \| Q_{X^n|MY^nZ^n} | P_{MZ^n}P_{Y^n|Z^n}) \\
    & \geq D(P_{MZ^n}P_{Y^n|Z^n} \| Q_{MY^nZ^n}),
\end{align}
where \eqref{eq:converse_proof_conditional_testing_against_dependence:use_conditoin_MXY_1} is due to the condition in \eqref{eq:converse_proof_conditional_testing_against_dependence:two_conditions_MXY_1}.
Hence, we get
\begin{align}
    & D(P_{MZ^n}P_{Y^n|Z^n} \| Q_{MY^nZ^n}) \nonumber \\
    & \leq D(\tilde{P}_{MX^nY^nZ^n} \| Q_{MX^nY^nZ^n}) \\
    & = D(P_{MX^nZ^n} \tilde{P}_{Y^n | MX^nZ^n} \| Q_{MX^nZ^n} Q_{Y^n|X^nZ^n} ) \label{eq:converse_proof_conditional_testing_against_dependence:use_conditoin_MXY_2} \\
    & = D(P_{MX^nZ^n} \| Q_{MX^nZ^n}) + D(\tilde{P}_{Y^n|MX^nZ^n} \| Q_{Y^n | X^nZ^n} | P_{MX^nZ^n})  \label{eq:converse_proof_conditional_testing_against_dependence:splitting_M_X} \\
    & = D(P_{X^nZ^n} \| Q_{X^nZ^n}) + D(\tilde{P}_{Y^n|MX^nZ^n} \| Q_{Y^n | X^nZ^n} | P_{MX^nZ^n}), \label{eq:converse_proof_conditional_testing_against_dependence:split_P_X_Q_X} 
\end{align}
where \eqref{eq:converse_proof_conditional_testing_against_dependence:use_conditoin_MXY_2} is due to the condition in \eqref{eq:converse_proof_conditional_testing_against_dependence:two_conditions_MXY_2}; \eqref{eq:converse_proof_conditional_testing_against_dependence:split_P_X_Q_X} is because \(P_{M|X^nZ^n} = Q_{M|X^nZ^n} = \idc\{M = \tilde{f}_n(X^n, Z^n)\}\).
Notice that \eqref{eq:converse_proof_conditional_testing_against_dependence:split_P_X_Q_X}  holds for any \(\tilde{P}_{Y^n|MX^nZ^n}\), as long as the joint distribution \(P_{MX^nZ^n}\tilde{P}_{Y^n|MX^nZ^n}\) satisfies \eqref{eq:converse_proof_conditional_testing_against_dependence:two_conditions_MXY_1}.
Thus, the task reduces to choosing a suitable \(\tilde{P}_{Y^n|MX^nZ^n}\) that can yield the desired single-letter upper bound.

We choose a suitable \(\tilde{P}_{Y^n|MX^nZ^n}\) in the same manner as in Section \ref{sec:converse_proof_thm_testing_against_dependence}.
Given the distribution \(P_{MX^nZ^n}\), we can obtain its marginal distribution \(P_{MX^{i}Z^n}\) for every \(i\in [n]\).
Let \(U_i \triangleq (M, X^{i-1}, Z^{i-1}, Z_{i+1}^{n})\) and hence \(P_{MX^{i}Z^n} = P_{U_i X_iZ_i}\).
For every \(i \in [n]\), we then select a \(\tilde{P}_{Y_i | U_i X_i Z_i}\) such that the joint distribution \(\tilde{P}_{U_i X_i Y_i Z_i} = P_{U_i X_i Z_i}\tilde{P}_{Y_i | U_i X_i Z_i}\) satisfies \(\tilde{P}_{U_i Y_i Z_i} = P_{U_iZ_i} P_{Y_i|Z_i}\), i.e., \(\tilde{P}_{Y_i|Z_i} = P_{Y_i|Z_i}\) and also \(U_i\) and \(Y_i\) are conditionally independent given \(Z_i\) under \(\tilde{P}_{U_i Y_i Z_i}\).
Such \(\tilde{P}_{Y_i | U_i X_i Z_i}\) must exist since we can at least choose \(\tilde{P}_{Y_i | U_i X_i Z_i} = P_{Y_i |Z_i}\).
Consider the \(\tilde{P}_{Y^n|MX^nZ^n}\) defined by
\begin{equation}
    \tilde{P}_{Y^n|MX^nZ^n} \triangleq \prod_{i=1}^{n} \tilde{P}_{Y_i | U_i X_i Z_i} = \prod_{i=1}^{n} \tilde{P}_{Y_i | M Z^nX^{i-1}X_i},  \label{eq:converse_proof_conditional_testing_against_dependence:define_suitable_product_distribution}
\end{equation}
where \(\tilde{P}_{Y_i | U_i X_i Z_i}\) is the one we just selected.
In order to substitute this \(\tilde{P}_{Y^n|MX^nZ^n}\) into \eqref{eq:converse_proof_conditional_testing_against_dependence:split_P_X_Q_X}, we need to verify that \(\tilde{P}_{MX^nY^nZ^n} = P_{MX^nZ^n}\tilde{P}_{Y^n|MX^nZ^n}\) satisfies \eqref{eq:converse_proof_conditional_testing_against_dependence:two_conditions_MXY_1}.
First, observe that
\begin{align}
    & \tilde{P}_{MY^nZ^n}  \nonumber \\
    & = \sum_{x^n} \tilde{P}_{MX^nY^nZ^n}\\
    & = P_{MZ^n}\sum_{x^n} P_{X^n | M Z^n} \tilde{P}_{Y^n | M X^n Z^n}  \\
    & = P_{MZ^n}\sum_{x^n} \prod_{i=1}^{n} P_{X_i | M Z^n X^{i-1}  }  \prod_{i=1}^{n}\tilde{P}_{Y_i | M Z^n X^{i-1} X_i }  \\
    & = P_{MZ^n}\sum_{x^n} \prod_{i=1}^{n} P_{X_i | M Z^n X^{i-1} }  \tilde{P}_{Y_i | M Z^n X^{i-1} X_i} \\
    & = P_{MZ^n} \sum_{x^{n-1}} \prod_{i=1}^{n-1} P_{X_i | M Z^n X^{i-1} }  \tilde{P}_{Y_i | M Z^n X^{i-1} X_i} \times \big( \sum_{x_n} P_{X_n | M Z^n X^{n-1} }  \tilde{P}_{Y_n | M Z^n X^{n-1} X_n}   \big). \label{eq:converse_proof_conditional_testing_against_dependence:sum_x_n}
\end{align}
For every fixed \((M, Z^n, X^{n-1}) = (m, z^n, x^{n-1})\), we have
\begin{equation}
    \sum_{x_n} P_{X_n | M Z^n X^{n-1} }  \tilde{P}_{Y_n | M Z^n X^{n-1} X_n}  = \tilde{P}_{Y_n | M Z^n X^{n-1}}.
\end{equation}
Recall that \(U_n = (M, X^{n-1}, Z^{n-1})\), and \(\tilde{P}_{Y_n | U_n X_n Z_n }\) is selected such that \(\tilde{P}_{U_nY_nZ_n} = P_{U_nZ_n}P_{Y_n |Z_n} \).
We then see that
\begin{align}
    \sum_{x_n} P_{X_n | M Z^n X^{n-1} }  \tilde{P}_{Y_n | M Z^n X^{n-1} X_n}  & = \tilde{P}_{Y_n | M Z^n X^{n-1}} \\
    & = \tilde{P}_{Y_n | U_n Z_n} \\
    & = P_{Y_n|Z_n}. \label{eq:converse_proof_conditional_testing_against_dependence:independence_Y_U}
\end{align}
Substituting \eqref{eq:converse_proof_conditional_testing_against_dependence:independence_Y_U} into \eqref{eq:converse_proof_conditional_testing_against_dependence:sum_x_n} yields
\begin{align}
     & \tilde{P}_{MZ^nY^n} \nonumber \\
     & = P_{MZ^n} \times P_{Y_n|Z_n} \times \sum_{x^{n-1}} \prod_{i=1}^{n-1} P_{X_i | M Z^n X^{i-1} } \tilde{P}_{Y_i | M Z^n X^{i-1} X_i}\\
     & = P_{MZ^n} \times P_{Y_n|Z_n} \times \sum_{x^{n-2}} \prod_{i=1}^{n-2} P_{X_i | M Z^n X^{i-1} }  \tilde{P}_{Y_i | M Z^n X^{i-1} X_i} \times \big( \sum_{x_{n-1}} P_{X_{n-1} | M Z^n X^{n-2} }  \tilde{P}_{Y_{n-1} | M Z^n X^{n-2} X_{n-1}}   \big).
\end{align}
By repeating the backward summation, we get
\begin{equation}
    \tilde{P}_{MZ^nY^n} = P_{MZ^n} \times \prod_{i=1}^{n}P_{Y_i |Z_i} = P_{MZ^n}P_{Y^n|Z^n}.
\end{equation}
Therefore, the \(\tilde{P}_{Y^n | M X^n Z^n}\) defined in \eqref{eq:converse_proof_conditional_testing_against_dependence:define_suitable_product_distribution} can be applied to the upper bound in \eqref{eq:converse_proof_conditional_testing_against_dependence:split_P_X_Q_X}.

By choosing this \(\tilde{P}_{Y^n | M X^n Z^n}\) in \eqref{eq:converse_proof_conditional_testing_against_dependence:split_P_X_Q_X}, we obtain that
\begin{align}
    & D(P_{MZ^n}P_{Y^n|Z^n} \| Q_{MZ^nY^n}) \nonumber \\
    & \leq D(P_{X^nZ^n} \| Q_{X^nZ^n}) + D(\tilde{P}_{Y^n|MX^nZ^n} \| Q_{Y^n | X^nZ^n} | P_{MX^nZ^n}) \\
    & = D(P_{X^nZ^n} \| Q_{X^nZ^n}) + \sum_{m, x^n, y^n, z^n} \tilde{P}_{MX^nY^nZ^n} \log \frac{ \tilde{P}_{Y^n|MX^nZ^n}  }{ Q_{Y^n|X^nZ^n} } \\
    & = D(P_{X^nZ^n} \| Q_{X^nZ^n}) + \sum_{m, x^n, y^n, z^n} \tilde{P}_{MX^nY^nZ^n} \log \frac{ \prod_{i=1}^{n} \tilde{P}_{Y_i | U_i X_i Z_i}  }{ \prod_{i=1}^{n} Q_{Y_i|X_iZ_i} } \\
    & = D(P_{X^nZ^n} \| Q_{X^nZ^n})  + \sum_{i=1}^{n} \sum_{u_i, x_i, y_i, z_i} \tilde{P}_{U_i X_i Y_i Z_i}  \log \frac{ \tilde{P}_{Y_i | U_i X_i Z_i }  }{ Q_{Y_i|X_i Z_i} } \\
    & = D(P_{X^nZ^n} \| Q_{X^nZ^n})  + \sum_{i=1}^{n} D(\tilde{P}_{Y_i | U_i X_i Z_i} \| Q_{Y_i|X_iZ_i} | P_{U_i X_i Z_i}) \\
    & = D(P_{X^nZ^n} \| Q_{X^nZ^n}) + n D(\tilde{P}_{Y_J | U_J X_J  Z_J J} \| Q_{Y_J|X_J Z_J J} | P_{U_J X_J  Z_J J}) \label{eq:converse_proof_conditional_testing_against_dependence:define_J}.\\
    & = D(P_{X^nZ^n} \| Q_{X^nZ^n}) + nD(\tilde{P}_{JU_JX_JY_JZ_J} \| P_{J U_J X_J Z_J} Q_{Y_J|X_J Z_J J}) \\
    & = D(P_{X^nZ^n} \| Q_{X^nZ^n}) + nD(\tilde{P}_{JU_JX_JY_JZ_J} \| P_{J U_J | X_J Z_J} P_{X_J Z_J} Q_{Y_J|X_J Z_J J}), \label{eq:converse_proof_conditional_testing_against_dependence:single_letterize_J}
\end{align}
where in \eqref{eq:converse_proof_conditional_testing_against_dependence:define_J}, \(J\)  is uniformly distributed over \([n]\), known as a time-sharing random variable.
Let \(\bar{U} \triangleq (J, U_J)\), and notice that \(P_{X_JZ_J} = P_{XZ}\), \(P_{Y_J} = P_{Y}\), and \(Q_{Y_J | X_J Z_J J} = Q_{Y|XZ}\) due to the i.i.d. assumption.
Thus, we can write  the upper bound in \eqref{eq:converse_proof_conditional_testing_against_dependence:single_letterize_J} as
\begin{align}
    & \frac{1}{n}D(P_{MZ^n}P_{Y^n|Z^n} \| Q_{MY^nZ^n}) \nonumber \\
    & \leq \frac{1}{n}D(P_{X^nZ^n} \| Q_{X^nZ^n}) + D(\tilde{P}_{\bar{U}XYZ} \| P_{\bar{U} | XZ} P_{XZ} Q_{Y|XZ}) \\
    & = D(P_{XZ} \| Q_{XZ}) + D(\tilde{P}_{\bar{U}XYZ} \| P_{\bar{U}| XZ} P_{XZ} Q_{Y|XZ}) \\
    & = D(\tilde{P}_{\bar{U}XYZ} \| P_{\bar{U} | XZ} Q_{XYZ}). \label{eq:converse_proof_conditional_testing_against_dependence:bound_in_U}
\end{align}
Since \(\tilde{P}_{U_i X_i Y_iZ_i} = P_{U_iX_iZ_i} \tilde{P}_{Y_i | U_i X_iZ_i}\) satisfies \(\tilde{P}_{U_iY_iZ_i} = P_{U_iZ_i}P_{Y_i|Z_i}\) for every \(i\), we see that \(\tilde{P}_{\bar{U}XYZ}  \) satisfies
\begin{align}
    \tilde{P}_{\bar{U}XZ} & = P_{\bar{U}XZ}, \label{eq:converse_proof_conditional_testing_against_dependence:marginal_condition_1} \\
    \tilde{P}_{\bar{U}YZ} & = P_{\bar{U}Z}P_{Y|Z}.\label{eq:converse_proof_conditional_testing_against_dependence:marginal_condition_2}
\end{align}
We have freedom in selecting \(\tilde{P}_{Y_i | U_i X_i Z_i}\) and hence \(\tilde{P}_{U_i X_i Y_iZ_i}\), so we can minimize over it for every \(i\) to obtain the tightest upper bound, yielding
\begin{align}
    \frac{1}{n}D(P_{MZ^n}P_{Y^n|Z^n} \| Q_{MY^nZ^n}) & \leq \min_{\tilde{P}_{\bar{U}XYZ}} D(\tilde{P}_{\bar{U}XYZ} \| P_{\bar{U}|XZ}Q_{XYZ}) \label{eq:converse_proof_conditional_testing_against_dependence:inner_minimization} \\
    & = \tilde{E}(P_{\bar{U}XZ}P_{Y|Z} \| Q_{\bar{U} XYZ}), \label{eq:converse_proof_conditional_testing_against_dependence:E_notation}
\end{align}
where in \eqref{eq:converse_proof_conditional_testing_against_dependence:inner_minimization}, \(\tilde{P}_{\bar{U}XYZ}\) needs to satisfy the conditions in \eqref{eq:converse_proof_conditional_testing_against_dependence:marginal_condition_1} and \eqref{eq:converse_proof_conditional_testing_against_dependence:marginal_condition_2}.

By the weak converse technique, we have
\begin{align}
    R & \geq \frac{1}{n} H(M) \\
    & \geq \frac{1}{n} H(M |Z^n) \\
    & \geq \frac{1}{n}I(M; X^n |Z^n) \\
    & = \frac{1}{n}\sum_{i=1}^{n}I(M;X_i|X^{i-1},Z^n) \\
    & = \frac{1}{n} \sum_{i=1}^{n}I(M, X^{i-1}, Z^{i-1}, Z_{i+1}^{n};X_i | Z_i) \label{eq:converse_proof_conditional_testing_against_dependence:weak_converse}\\
    & = \frac{1}{n} \sum_{i=1}^{n}I(U_i;X_i |Z_i)  \\
    & = I( U_J ; X_J | Z_J, J) \\
    & = I(U_J, J; X_J | Z_J) \\
    & = I(\bar{U}; X |Z),
\end{align}
where \eqref{eq:converse_proof_conditional_testing_against_dependence:weak_converse} is because \(H(X_i|X^{i-1}, Z^n) = H(X_i|Z_i)\).
Therefore, we can upper bound \eqref{eq:converse_proof_conditional_testing_against_dependence:E_notation} by maximizing over all \(U\) satisfying \(I(U;X |Z) \leq R\) and conclude that
\begin{align}
    \frac{1}{n}D(P_{MZ^n}P_{Y^n|Z^n} \| Q_{MY^nZ^n}) & \leq \max_{ \substack{P_{U|XZ}: \\ I(U;X |Z ) \leq R} } \tilde{E}(P_{UXZ}P_{Y|Z} \| Q_{UXYZ}).
\end{align}
We prove the cardinality bound \(|\mathcal{U}| \leq |\mathcal{X}||\mathcal{Z}| + 1 \) in Appendix \ref{apd:cardinality_bound_conditional_testing_against_dependence}.
Our proof is different from that of \cite[Theorem 3]{hanHypothesisTestingMultiterminal1987} owing to new observations, see Remark \ref{rem:cardinality_bound}.
With this, the converse proof is completed.

\section{Proof of Theorem \ref{thm:upper_bound_general_setting}}
\label{apd:upper_bound_general_setting}
We prove Theorem \ref{thm:upper_bound_general_setting} in this section.
That is, we adapt the proof for testing against dependence to derive an upper bound for \(E(R)\) under the general setting, where the receiver decides between \(\mathcal{H}_0:P_{XY}\) and \(\mathcal{H}_1:Q_{XY}\).
For completeness, we repeat the initial part of that proof in the general setting.
Consider a sequence of testing schemes \((f_n, \varphi_n)\) with \(||f_n|| \ndot{\leq} e^{nR}\), and denote by \(E\) the achievable error exponent under the sequence.
By applying the schemes to the two hypotheses respectively, we obtain two joint distributions \(P_{MX^nY^n}\) and \(Q_{MX^nY^n}\).
Following the multi-letter characterization proof \cite[Theorem 1]{ahlswedeHypothesisTestingCommunication1986}, we see that
\begin{equation}
    E \leq \liminf_{n\to \infty} \frac{1}{n} D(P_{MY^n} \| Q_{MY^n}).  \label{eq:apd_upper_bound_general_setting:multi-letter_characterization}
\end{equation}
The goal is to derive a single-letter upper bound for \eqref{eq:apd_upper_bound_general_setting:multi-letter_characterization}

To this end, consider a distribution \(\tilde{P}_{MX^nY^n}\) that satisfies the following conditions
\begin{align}
    \tilde{P}_{MY^n} & = P_{MY^n}, \label{eq:apd_upper_bound_general_setting:two_conditions_MXY_1} \\
    \tilde{P}_{MX^n} & = P_{MX^n}. \label{eq:apd_upper_bound_general_setting:two_conditions_MXY_2}
\end{align}
The reasons for imposing the two conditions will become clear shortly.
Notice that
\begin{align}
    & D(\tilde{P}_{MX^nY^n} \| Q_{MX^nY^n}) \nonumber \\
    & = D(P_{MY^n} \tilde{P}_{X^n | M Y^n} \| Q_{MY^n} Q_{X^n|MY^n}) \label{eq:apd_upper_bound_general_setting:use_conditoin_MXY_1} \\
    & = D(P_{MY^n} \| Q_{MY^n}) + D( \tilde{P}_{X^n | M Y^n} \| Q_{X^n|MY^n} | P_{MY^n}) \\
    & \geq D(P_{MY^n} \| Q_{MY^n}),
\end{align}
where \eqref{eq:apd_upper_bound_general_setting:use_conditoin_MXY_1} is due to the condition in \eqref{eq:apd_upper_bound_general_setting:two_conditions_MXY_1}.
Hence, we get
\begin{align}
    & D(P_{MY^n} \| Q_{MY^n}) \nonumber \\
    & \leq D(\tilde{P}_{MX^nY^n} \| Q_{MX^nY^n}) \\
    & = D(P_{MX^n} \tilde{P}_{Y^n | MX^n} \| Q_{MX^n} Q_{Y^n|X^n} ) \label{eq:apd_upper_bound_general_setting:use_conditoin_MXY_2} \\
    & = D(P_{MX^n} \| Q_{MX^n}) + D(\tilde{P}_{Y^n|MX^n} \| Q_{Y^n | X^n} | P_{MX^n})  \label{eq:apd_upper_bound_general_setting:splitting_M_X} \\
    & = D(P_{X^n} \| Q_{X^n}) + D(\tilde{P}_{Y^n|MX^n} \| Q_{Y^n | X^n} | P_{MX^n}), \label{eq:apd_upper_bound_general_setting:split_P_X_Q_X} 
\end{align}
where \eqref{eq:apd_upper_bound_general_setting:use_conditoin_MXY_2} is due to the condition in \eqref{eq:apd_upper_bound_general_setting:two_conditions_MXY_2}; \eqref{eq:apd_upper_bound_general_setting:split_P_X_Q_X} is because \(P_{M|X^n} = Q_{M|X^n} = \idc\{M = f_n(X^n)\}\).
Notice that \eqref{eq:apd_upper_bound_general_setting:split_P_X_Q_X}  holds for any \(\tilde{P}_{Y^n|MX^n}\), as long as the joint distribution \(P_{MX^n}\tilde{P}_{Y^n|MX^n}\) satisfies \eqref{eq:apd_upper_bound_general_setting:two_conditions_MXY_1}.
Thus, the task reduces to choosing a suitable \(\tilde{P}_{Y^n|MX^n}\) that can yield the desired single-letter upper bound.

We choose a suitable \(\tilde{P}_{Y^n|MX^n}\) similarly to Section \ref{sec:converse_proof_thm_testing_against_dependence}, but use a different auxiliary random variable.
Given the distribution \(P_{MX^nY^n}\), we can obtain its marginal distribution \(P_{MX^{i}Y^{n}_{i+1}}\) for every \(i\in [n]\).
Let \(U_i \triangleq (M, X^{i-1},Y_{i+1}^n)\) and hence \(P_{MX^{i}Y_{i+1}^n} = P_{U_i X_i}\).
For every \(i \in [n]\), we then select a \(\tilde{P}_{Y_i | U_i X_i}\) such that the joint distribution \(\tilde{P}_{U_i X_i Y_i} = P_{U_i X_i}\tilde{P}_{Y_i | U_i X_i}\) satisfies \(\tilde{P}_{U_i Y_i} = P_{U_iY_i}\), i.e. \(\tilde{P}_{Y_i} = P_{Y_i}\) and \(\tilde{P}_{Y_i|U_i} = P_{Y_i|U_i}\) where \(P_{Y_i|U_i}\) is determined by \(P_{MX^nY^n}\).
Such \(\tilde{P}_{Y_i | U_i X_i}\) must exist since we can at least choose \(\tilde{P}_{Y_i | U_i X_i} = P_{Y_i|X_i}\), which satisfies \(\tilde{P}_{U_iY_i} = P_{U_iY_i}\) due to the Markov chain \(U_i \to X_i \to Y_i\).
Consider the \(\tilde{P}_{Y^n|MX^n}\) defined by
\begin{equation}
    \tilde{P}_{Y^n|MX^n} \triangleq \prod_{i=1}^{n} \tilde{P}_{Y_i | U_i X_i} = \prod_{i=1}^{n} \tilde{P}_{Y_i | M X^{i-1}Y_{i+1}^{n}X_i},  \label{eq:apd_upper_bound_general_setting:define_suitable_product_distribution}
\end{equation}
where \(\tilde{P}_{Y_i | U_i X_i}\) is the one we just selected.
In order to substitute this \(\tilde{P}_{Y^n|MX^n}\) into \eqref{eq:apd_upper_bound_general_setting:split_P_X_Q_X}, we need to verify that \(\tilde{P}_{MX^nY^n} = P_{MX^n}\tilde{P}_{Y^n|MX^n}\) satisfies \eqref{eq:apd_upper_bound_general_setting:two_conditions_MXY_1}.
This requires a more sophisticated backward summation argument.
First, observe that
\begin{align}
    & \tilde{P}_{MY^n}  \nonumber \\
    & = \sum_{x^n} \tilde{P}_{MX^nY^n}\\
    & = P_{M}\sum_{x^n} P_{X^n | M } \tilde{P}_{Y^n | M X^n}  \\
    & = P_{M}\sum_{x^n} \prod_{i=1}^{n} P_{X_i | M X^{i-1} }  \prod_{i=1}^{n}\tilde{P}_{Y_i | M X^{i-1}Y_{i+1}^n X_i}  \\
    & = P_{M}\sum_{x^n} \prod_{i=1}^{n} P_{X_i | M X^{i-1} }  \tilde{P}_{Y_i | M X^{i-1} Y_{i+1}^nX_i} \\
    & = P_{M} \sum_{x^{n-1}} \prod_{i=1}^{n-1} P_{X_i | M X^{i-1} }  \tilde{P}_{Y_i | M X^{i-1} Y_{i+1}^n X_i} \times \big( \sum_{x_n} P_{X_n | M X^{n-1} }  \tilde{P}_{Y_n | M X^{n-1}  X_n}  \big) \\
    & = P_{M} \sum_{x^{n-1}} \prod_{i=1}^{n-1} P_{X_i | M X^{i-1} }  \tilde{P}_{Y_i | M X^{i-1} Y_{i+1}^n X_i} \times \tilde{P}_{Y_n | M X^{n-1} } \label{eq:apd_upper_bound_general_setting:sum_x_n} \\
    & = P_{M} \sum_{x^{n-1}} \prod_{i=1}^{n-1} P_{X_i | M X^{i-1} }  \tilde{P}_{Y_i | M X^{i-1} Y_{i+1}^n X_i} \times P_{Y_n | M X^{n-1} } \label{eq:apd_upper_bound_general_setting:replace_y_n} \\
     & = P_{M} \sum_{x^{n-2}} \prod_{i=1}^{n-2} P_{X_i | M X^{i-1} }\tilde{P}_{Y_i | M X^{i-1} Y_{i+1}^n X_i} \times \big( \sum_{x_{n-1}} P_{X_{n-1} | M X^{n-2} }  \tilde{P}_{Y_{n-1} | M X^{n-2} Y_{n} X_{n-1}}  \times P_{Y_n | M X^{n-1} } \big),  \label{eq:apd_upper_bound_general_setting:moving_x_n_1_forward} 
\end{align}
where \eqref{eq:apd_upper_bound_general_setting:sum_x_n}  is because
\begin{equation}
    \sum_{x_n} P_{X_n | M X^{n-1} }  \tilde{P}_{Y_n | M X^{n-1}  X_n} =   \tilde{P}_{Y_n | M X^{n-1} };
\end{equation}
\eqref{eq:apd_upper_bound_general_setting:replace_y_n} is because \(U_n = (M, X^{n-1})\) and hence we have
\begin{equation}
    \tilde{P}_{Y_n | M X^{n-1} } = \tilde{P}_{Y_n | U_n } = P_{Y_n|U_n}.
\end{equation}

We now take a closer look at the inner summation over \(x_{n-1}\) in \eqref{eq:apd_upper_bound_general_setting:moving_x_n_1_forward}.
First, observe that 
\begin{align}
    P_{X_{n-1} | M X^{n-2} } P_{Y_n | M X^{n-1} } & = P_{X_{n-1} | M X^{n-2} } P_{Y_n | M X^{n-2} X_{n-1}} \\
    & = P_{Y_n X_{n-1}|M X^{n-2}} \\
    & = P_{Y_n|M X^{n-2}} P_{X_{n-1}|M X^{n-2}Y_n}.
\end{align}
It follows that
\begin{align}
    & \sum_{x_{n-1}} P_{X_{n-1} | M X^{n-2} }  \tilde{P}_{Y_{n-1} | M X^{n-2} Y_{n} X_{n-1}}  \times P_{Y_n | M X^{n-1} } \nonumber \\
    & =\sum_{x_{n-1}} P_{Y_n|M X^{n-2}} P_{X_{n-1}|M X^{n-2}Y_n} \tilde{P}_{Y_{n-1} | M X^{n-2} Y_{n} X_{n-1}} \\
    & = P_{Y_n|M X^{n-2}} \sum_{x_{n-1}}  P_{X_{n-1}|M X^{n-2}Y_n}\tilde{P}_{Y_{n-1} | M X^{n-2} Y_{n} X_{n-1}} \\
    & = P_{Y_n|M X^{n-2}} \tilde{P}_{Y_{n-1} | M X^{n-2} Y_{n}} \\
    & = P_{Y_n|M X^{n-2}} P_{Y_{n-1} | M X^{n-2} Y_{n}} \label{eq:apd_upper_bound_general_setting:replace_Y_n_1} \\
    & = P_{Y_{n}, Y_{n-1}|M X^{n-2}}, \label{eq:apd_upper_bound_general_setting:combining_Y_n_Y_n_1}
\end{align}
where \eqref{eq:apd_upper_bound_general_setting:replace_Y_n_1} is because \(U_{n-1} = (M, X^{n-2}, Y_n)\) and
\begin{equation}
    \tilde{P}_{Y_{n-1}|MX^{n-2} Y_n} = \tilde{P}_{Y_{n-1}|U_{n-1}} = P_{Y_{n-1}|U_{n-1}}.
\end{equation}
Substituting \eqref{eq:apd_upper_bound_general_setting:combining_Y_n_Y_n_1} into  \eqref{eq:apd_upper_bound_general_setting:moving_x_n_1_forward} yields 
\begin{align}
    \tilde{P}_{MY^n} = P_{M} \sum_{x^{n-2}} \prod_{i=1}^{n-2} P_{X_i | M X^{i-1} }\tilde{P}_{Y_i | M X^{i-1} Y_{i+1}^n X_i}  \times P_{Y_{n-1}, Y_n|MX^{n-2}}. \label{eq:apd_upper_bound_general_setting:using_combining_Y_n_1_and_Y_n}
\end{align}
Comparing  \eqref{eq:apd_upper_bound_general_setting:using_combining_Y_n_1_and_Y_n} and \eqref{eq:apd_upper_bound_general_setting:replace_y_n}, we see that by repeating the backward summation argument, we get
\begin{equation}
    \tilde{P}_{MY^n} = P_{M} \times P_{Y^n|M} = P_{MY^n}.
\end{equation}
Therefore, the \(\tilde{P}_{Y_n | M X^n}\) defined in \eqref{eq:apd_upper_bound_general_setting:define_suitable_product_distribution} can be applied to the upper bound in \eqref{eq:apd_upper_bound_general_setting:split_P_X_Q_X}.

Applying the \(\tilde{P}_{Y_n | M X^n}\) defined in \eqref{eq:apd_upper_bound_general_setting:define_suitable_product_distribution}  to \eqref{eq:apd_upper_bound_general_setting:split_P_X_Q_X}, we obtain that
\begin{align}
    & D(P_{MY^n} \| Q_{MY^n}) \nonumber \\
    & \leq D(P_{X^n} \| Q_{X^n}) + D(\tilde{P}_{Y^n|MX^n} \| Q_{Y^n | X^n} | P_{MX^n}) \\
    & = D(P_{X^n} \| Q_{X^n}) + \sum_{m, x^n, y^n} \tilde{P}_{MX^nY^n} \log \frac{ \tilde{P}_{Y^n|MX^n}  }{ Q_{Y^n|X^n}} \\
    & = D(P_{X^n} \| Q_{X^n}) + \sum_{m, x^n, y^n} \tilde{P}_{MX^nY^n} \log \frac{ \prod_{i=1}^{n} \tilde{P}_{Y_i | U_i X_i} }{ \prod_{i=1}^{n} Q_{Y_i|X_i}} \\
    & = D(P_{X^n} \| Q_{X^n})  + \sum_{i=1}^{n} \sum_{u_i, x_i, y_i} \tilde{P}_{U_i X_i Y_i} \log \frac{ \tilde{P}_{Y_i | U_i X_i}  }{ Q_{Y_i|X_i} }. \label{eq:apd_upper_bound_general_setting:exchange_sum}
\end{align}
To proceed, we need to verify that \(\tilde{P}_{U_iX_iY_i}\) satisfies \(\tilde{P}_{U_iX_i} = P_{U_iX_i}\).
We show this as follows.
Recall that \(U_i = (M, X^{i-1}, Y_{i+1}^n)\). Then, for every \(k \in [n]\), we have
\begin{align}
    \tilde{P}_{U_kX_k} & = \tilde{P}_{M X^{k}Y_{k+1}^n} \\
    & =\sum_{y^{k}, x_{k+1}^{n}} \tilde{P}_{M X^n Y^n} \\
    & = \sum_{y^{k}}  P_{M}  \prod_{i=1}^{k} P_{X_i | M X^{i-1} }\tilde{P}_{Y_i | M X^{i-1} Y_{i+1}^n X_i}  \times P_{Y_{k+1}^n|MX^{k}} \label{eq:apd_upper_bound_general_setting:applying_backward_summation_to_k}  \\
    & =  P_{M}  \prod_{i=1}^{k} P_{X_i | M X^{i-1} } \times P_{Y_{k+1}^n|MX^{k}} \label{eq:apd_upper_bound_general_setting:remove_y_k} \\
    & = P_{M}P_{X^k|M}P_{Y_{k+1}^n|MX^{k}} \\
    & = P_{M X^k Y_{k+1}^n} \\
    & = P_{U_kX_k}
\end{align}
where in \eqref{eq:apd_upper_bound_general_setting:applying_backward_summation_to_k}, we apply the backward summation argument used in deriving  \eqref{eq:apd_upper_bound_general_setting:using_combining_Y_n_1_and_Y_n}, i.e., we sum only from \(i=n\) to \(k+1\);
in \eqref{eq:apd_upper_bound_general_setting:remove_y_k}, we note that the remaining distributions do not depend on \(y^k\).

We now continue the derivation from \eqref{eq:apd_upper_bound_general_setting:exchange_sum} and obtain that
\begin{align}
    & D(P_{MY^n} \| Q_{MY^n}) \nonumber \\
    & \leq D(P_{X^n} \| Q_{X^n})  + \sum_{i=1}^{n} D(\tilde{P}_{Y_i | U_i X_i} \| Q_{Y_i|X_i} | P_{U_i X_i}) \\
    & = D(P_{X^n} \| Q_{X^n}) + n D(\tilde{P}_{Y_J | U_J X_J J} \| Q_{Y_J|X_J J} | P_{U_J X_J J}) \label{eq:apd_upper_bound_general_setting:define_J}.\\
    & = D(P_{X^n} \| Q_{X^n}) + nD(\tilde{P}_{JU_JX_JY_J} \| P_{J U_J X_J} Q_{Y_J|X_J J}) \\
    & = D(P_{X^n} \| Q_{X^n}) + nD(\tilde{P}_{JU_JX_JY_J} \| P_{J U_J | X_J} P_{X_J} Q_{Y_J|X_J J}), \label{eq:apd_upper_bound_general_setting:single_letterize_J}
\end{align}
where in \eqref{eq:apd_upper_bound_general_setting:define_J}, \(J\)  is uniformly distributed over \([n]\), known as a time-sharing random variable.
Let \(\bar{U} \triangleq (J, U_J)\), and notice that \(P_{X_J} = P_{X}\), \(P_{Y_J} = P_{Y}\), and \(Q_{Y_J | X_J, J} = Q_{Y|X}\) due to the i.i.d. assumption.
Thus, we can write  the upper bound in \eqref{eq:apd_upper_bound_general_setting:single_letterize_J} as
\begin{align}
    & \frac{1}{n}D(P_{MY^n} \| Q_{MY^n}) \nonumber \\
    & \leq \frac{1}{n}D(P_{X^n} \| Q_{X^n}) + D(\tilde{P}_{\bar{U}XY} \| P_{\bar{U} | X} P_{X} Q_{Y|X}) \\
    & = D(P_{X} \| Q_{X}) + D(\tilde{P}_{\bar{U}XY} \| P_{\bar{U}| X} P_{X} Q_{Y|X}) \\
    & = D(\tilde{P}_{\bar{U}XY} \| P_{\bar{U} | X} Q_{XY}). \label{eq:apd_upper_bound_general_setting:bound_in_U}
\end{align}
Since \(\tilde{P}_{U_i X_i Y_i} = P_{U_iX_i} \tilde{P}_{Y_i | U_i X_i}\) satisfies \(\tilde{P}_{U_iY_i} = P_{U_iY_i}\) for every \(i\), we see that \(\tilde{P}_{\bar{U}XY}  \) satisfies
\begin{align}
    \tilde{P}_{\bar{U}X} & = P_{\bar{U}X}, \label{eq:apd_upper_bound_general_setting:marginal_condition_1} \\
    \tilde{P}_{\bar{U}Y} & = P_{\bar{U}Y}.\label{eq:apd_upper_bound_general_setting:marginal_condition_2}
\end{align}
We have freedom in selecting \(\tilde{P}_{Y_i | U_i X_i}\) and hence \(\tilde{P}_{U_iX_iY_i}\), so we can minimize over it for every \(i\) to obtain the tightest upper bound, yielding
\begin{align}
    \frac{1}{n}D(P_{MY^n} \| Q_{MY^n}) & \leq \min_{\tilde{P}_{\bar{U}XY}} D(\tilde{P}_{\bar{U}XY} \| P_{\bar{U}|X}Q_{XY}) \label{eq:apd_upper_bound_general_setting:inner_minimization} \\
    & = E(P_{\bar{U}XY} \| Q_{\bar{U} XY}), \label{eq:apd_upper_bound_general_setting:E_notation}
\end{align}
where in \eqref{eq:apd_upper_bound_general_setting:inner_minimization}, \(\tilde{P}_{\bar{U}XY}\) needs to satisfy the conditions in \eqref{eq:apd_upper_bound_general_setting:marginal_condition_1} and \eqref{eq:apd_upper_bound_general_setting:marginal_condition_2}.

By the weak converse technique, we have
\begin{align}
    R & \geq \frac{1}{n} H(M) \\
    & \geq \frac{1}{n} H(M|Y^n) \\
    & \geq \frac{1}{n}I(M; X^n|Y^n) \\
    & = \frac{1}{n}\sum_{i=1}^{n}I(M;X_i|X^{i-1}, Y^n) \\
    & = \frac{1}{n}\sum_{i=1}^{n} \big(H(X_i|X^{i-1}, Y^n) - H(X_i|M, X^{i-1}, Y^n) \big) \\
    & = \frac{1}{n}\sum_{i=1}^{n} \big(H(X_i|Y_i) - H(X_i|M, X^{i-1}, Y^{n}_{i+1}, Y_i) \big) \label{eq:apd_upper_bound_general_setting:weak_converse_simplification} \\
    & = \frac{1}{n} \sum_{i=1}^{n}I(M, X^{i-1}, Y^{n}_{i+1}; X_i |Y_i) \\
    & = \frac{1}{n} \sum_{i=1}^{n}I(U_i;X_i|Y_i)  \\
    & = I( U_J ; X_J | Y_J, J) \\
    & = I(U_J, J; X_J |Y_J) \\
    & = I(\bar{U}; X |Y).
\end{align}
where \eqref{eq:apd_upper_bound_general_setting:weak_converse_simplification} follows from the Markov chain \(Y^{i-1} \to (M, X^{i-1}, Y^{n}_{i+1}, Y_i) \to X_i\).
Therefore, we can upper bound \eqref{eq:apd_upper_bound_general_setting:E_notation} by maximizing over all \(U\) satisfying \(I(U;X|Y) \leq R\) and conclude that
\begin{align}
    \frac{1}{n}D(P_{MY^n} \| Q_{MY^n}) & \leq \max_{ \substack{P_{U|X}: \\ I(U;X|Y) \leq R} } E(P_{UXY} \| Q_{UXY}).
\end{align}
The cardinality bound \(|\mathcal U| \le |\mathcal X|+1\) follows by modifying the proof in Appendix \ref{apd:cardinality_bound_conditional_testing_against_dependence}, in particular by choosing \(Z = \emptyset\) and 
\begin{equation}
    g_0(P_{X|U=u}) \triangleq H(X|Y, U=u),
\end{equation}
where the latter choice accounts for the constraint \(I(U;X|Y)\le R\).
Note that if we directly apply Han's proof of the cardinality bound, we can only obtain \(|\mathcal U| \le |\mathcal X|+2\).
With this, the converse proof is completed.

\section{Concluding Remarks}
\label{sec:concluding_remarks}
In this paper, we studied testing against dependence and showed that \(E(R)\) in this setting equals Han's exponent.
We derived the matching converse bound by single-letterizing a multi-letter version of Han's exponent, rather than directly single-letterizing Ahlswede-Csiszár's exponent.
A key step is to ensure that the selected auxiliary distribution satisfies the required marginal constraints, which is done by a backward summation argument.
We also disproved Han's conjecture by showing that Han's exponent can be strictly larger than the one given by the lautum information.
In the binary case, we characterized the conditions under which the two coincide.
We then studied testing against the product of dependence and independence.
We showed that \(E(R)\) in this setting also equals Han's exponent and has a simpler expression.
This result demonstrated the optimality of Han’s exponent beyond the dependence and independence settings.
We next studied testing against conditional dependence by considering the related problem of conditional testing against dependence.
We derived a single-letter characterization of \(E(R)\) for the latter problem, which in turn provides a single-letter bound for the former.
When \(Z\) is a function of \(X\), we showed that the bound is tight by using the conditional quantization scheme.
The optimal exponents in this paper uncover several settings in which binning is not needed.
Lastly, we extended the converse method developed in this paper to the general setting, and recovered the same single-letter converse bound recently established in \cite[Corollary 3]{wenExactRateExponent2026}.
For future work, it may be of interest to study multiterminal variants of testing against dependence, or extensions with constraints such as privacy.

\appendices

\section{Proof of Proposition \ref{prop:disproving_Han_conjecture}}
\label{apd:proof_disproving_Han_conjecture}
We prove Proposition \ref{prop:disproving_Han_conjecture} in this appendix.
We first focus on the if and only if statement.
The if part is immediate, so it remains to show that, given \(|\mathcal{X}| = |\mathcal{Y}|=2\) and \(0 \leq R \leq H(X) \),  \(E_{\mathrm{H}}(R) = L(R)\)  implies
\begin{equation}
    E_{\mathrm{H}}(R) = L(R) = \frac{R}{H(X)} D(P_{X}P_{Y} \| P_{XY}).
\end{equation}
Before starting the proof, let us first consider the special instance \(P_{XY} = P_{X}P_{Y}\). It is easy to verify that \(E_{\mathrm{H}}(R) = L(R) = 0\).
Moreover, \(D(P_{X}P_{Y}\|P_{XY}) = 0\).
Hence, the if and only if statement holds for this instance.
We therefore exclude it from the rest of the proof.

Recall the definition of \(E(P_{UX}P_{Y} \| P_{UXY})\) in \eqref{eq:problem_setup:Han_exponent_definition_1} and \eqref{eq:problem_setup:Han_exponent_definition_2}.
It follows that
\begin{align}
    E(P_{UX}P_{Y} \| P_{UXY}) & = \min_{\substack{\tilde{P}_{UXY} \in \mathcal{P}_{\mathrm{H}}(P_{UX}P_{Y})  }} D( \tilde{P}_{UXY} \| P_{UXY} ) \\
    & =  D(P_{U}P_{Y} \| P_{UY}) + \min_{\substack{\tilde{P}_{UXY} \in \mathcal{P}_{\mathrm{H}}(P_{UX}P_{Y})  }} D( \tilde{P}_{X|UY} \| P_{X|UY} | P_{U}P_{Y}) \label{eq:apd:disproving_Han_conjecture:lautum_lower_bound} \\
    & \geq  D(P_{U}P_{Y} \| P_{UY}) \\
    & = L(U;Y),
\end{align}
where \eqref{eq:apd:disproving_Han_conjecture:lautum_lower_bound} is because \(\tilde{P}_{UXY} \in \mathcal{P}_{\mathrm{H}}(P_{UX}P_{Y}) \) means that we can write \(\tilde{P}_{UXY} = P_{U}P_{Y}\tilde{P}_{X|UY}\).
Hence, \(E(P_{UX}P_{Y} \| P_{UXY}) \geq L(U;Y)\) holds for every admissible \(P_{U|X}\).
Recall that
\begin{align}
     E_{\mathrm{H}}(R) & = \max_{\substack{P_{U|X}: \\ I(U;X) \leq R } }  E(P_{UX}P_{Y} \| P_{UXY}), \\
    L(R) & = \max_{\substack{P_{U|X}: \\ I(U;X) \leq R } } L(U;Y). 
\end{align}
Thus,  \(E_{\mathrm{H}}(R) = L(R)\) implies that both exponents are achieved by a common maximizer \(P_{ U | X}\) satisfying
\begin{equation}
    \min_{\substack{\tilde{P}_{UXY} \in \mathcal{P}_{\mathrm{H}}(P_{UX}P_{Y})  }} D( \tilde{P}_{X|UY} \| P_{X|UY} | P_{U}P_{Y}) = 0. \label{eq:apd:disproving_Han_conjecture:0_divergence}
\end{equation}
Notice that \eqref{eq:apd:disproving_Han_conjecture:0_divergence} is true only if \(\tilde{P}_{UXY} = P_{U}P_{Y}P_{X|UY} \) belongs to \( \mathcal{P}_{\mathrm{H}}(P_{UX}P_{Y}) \), i.e., for every \((u,x)\), we have
\begin{equation}
    \sum_{y}P_{U}(u)P_{Y}(y)P_{X|UY}(x|u,y) = P_{UX}(u,x).  \label{eq:apd:disproving_Han_conjecture:sum_y_marginal_1}
\end{equation}
On the other hand, from \(P_{UXY} = P_{U|X}P_{XY} = P_{UY}P_{X|UY}\), we get
\begin{equation}
    P_{UX}(u,x) = \sum_{y} P_{UY}(u,y)P_{X|UY}(x|u,y).\label{eq:apd:disproving_Han_conjecture:sum_y_marginal_2}
\end{equation}
Combining \eqref{eq:apd:disproving_Han_conjecture:sum_y_marginal_1} and \eqref{eq:apd:disproving_Han_conjecture:sum_y_marginal_2}, we see that the maximizer \(P_{U|X}\) must satisfy
\begin{equation}
    \sum_{y}P_{U}(u)P_{Y}(y)P_{X|UY}(x|u,y) = \sum_{y} P_{UY}(u,y)P_{X|UY}(x|u,y),
\end{equation}
or equivalently,
\begin{equation}
    \sum_{y} \big(P_{Y}(y)-P_{Y|U}(y|u)\big)P_{X|UY}(x|u,y)=0, \label{eq:apd:disproving_Han_conjecture:sum_y_marginal_combine}
\end{equation}
which needs to hold for every \((u,x)\) in the support of \(P_{UX}\).

\subsection{Conditions for the Maximizer}
We first study what \eqref{eq:apd:disproving_Han_conjecture:sum_y_marginal_combine} tells us about the maximizer \(P_{U|X}\).
To proceed, let us first observe that for every \(u\), we have
\begin{equation}
    \sum_{y} \big(P_{Y}(y)-P_{Y|U}(y|u)\big)=0. \label{eq:apd:disproving_Han_conjecture:sum_y_marginal_no_X}
\end{equation}
Recall the assumption \(|\mathcal{Y}| = 2\) and let \(\mathcal{Y} = \{0,1\}\) without loss of generality.
For a fixed \((u,x)\), let
\begin{align}
    a_0 & \triangleq P_{Y}(0) - P_{Y|U}(0|u), \\
    a_1 & \triangleq P_{Y}(1) - P_{Y|U}(1|u),
\end{align}
as well as \(b_0 \triangleq P_{X|UY}(x|u,0)\) and \(b_1 \triangleq P_{X|UY}(x|u,1)\).
By \eqref{eq:apd:disproving_Han_conjecture:sum_y_marginal_no_X} and \eqref{eq:apd:disproving_Han_conjecture:sum_y_marginal_combine}, we have
\begin{align}
    a_0 + a_1 &=0,  \label{eq:apd:disproving_Han_conjecture:a_b_equation_1}   \\ 
    a_0b_0 + a_1b_1 & =0. \label{eq:apd:disproving_Han_conjecture:a_b_equation_2} 
\end{align}
It can be verified that \eqref{eq:apd:disproving_Han_conjecture:a_b_equation_1} and \eqref{eq:apd:disproving_Han_conjecture:a_b_equation_2} hold in either of the following two cases: 
\begin{enumerate}
    \item \(a_0=a_1 =0\), while \(b_0\) and \(b_1\) can be arbitrary.
    \item \(a_0=-a_1 \neq 0\) and \(b_0=b_1\).
\end{enumerate}
Since \(a_0=a_1 =0\) means \(P_{Y} = P_{Y|U=u}\), we see that the two cases correspond to whether \(P_{Y} = P_{Y|U=u}\) holds or not.
To this end, we partition \(\mathcal{U}\) into two subsets \(\mathcal{U}_1\) and \(\mathcal{U}_2\), where
\begin{align}
    P_{Y} & = P_{Y|U=u}, \qquad \forall u \in \mathcal{U}_1, \\
    P_{Y} & \neq P_{Y|U=u}, \qquad \forall u \in \mathcal{U}_2.
\end{align}
If \(u \in \mathcal{U}_2\), i.e., the second case is true, then \eqref{eq:apd:disproving_Han_conjecture:a_b_equation_1} and \eqref{eq:apd:disproving_Han_conjecture:a_b_equation_2} tell us \(P_{X|UY}(x|u,0) = P_{X|UY}(x|u,1)\).
Since
\begin{equation}
    P_{X|U}(x|u) = \sum_{y}P_{Y|U}(y|u) P_{X|UY}(x|u,y),
\end{equation}
we can verify that \(P_{X|UY}(x|u,0) = P_{X|UY}(x|u,1)\) implies \(P_{X|U=u} = P_{X|U=u,Y}\), i.e., \(X\) and \(Y\) are conditionally independent given \(U = u\).

In summary, \eqref{eq:apd:disproving_Han_conjecture:sum_y_marginal_combine} implies that based on the maximizer \(P_{U|X}\), we can  partition \(\mathcal{U}\) into two subsets \(\mathcal{U}_1\) and \(\mathcal{U}_2\) such that
\begin{enumerate}
    \item if \(u \in \mathcal{U}_1\), then \(P_{Y|U=u} = P_{Y}\);
    \item if \(u \in \mathcal{U}_2\), then \(P_{Y|U=u} \neq P_{Y}\)  but \(P_{X|U = u} = P_{X|U=u, Y} \).
\end{enumerate}

\subsection{Double Markovity}
We now take a closer look at the second case \(u \in \mathcal{U}_2\).
In what follows, we assume that the maximizer \(P_{U|X}\) satisfies \(\mathcal{U}_1 = \emptyset\), i.e., there is no \(u \in \mathcal{U}\) such that \(P_{Y|U=u} = P_{Y}\).
We will return to the case where \(\mathcal{U}_1 \neq \emptyset\) further on.
Since \(\mathcal{U}_1 = \emptyset\), we see that \(P_{X|U} = P_{X|UY}\).
Thus, under the joint distribution \(P_{UXY} = P_{U|X}P_{XY}\), the following two Markov chains hold simultaneously
\begin{equation}
    U \to X \to Y \quad \text{and} \quad X \to U \to Y. \label{eq:apd:proof_disproving_Han_conjecture:double_markovity:double_chains}
\end{equation}
To proceed, let us recall the following known result in the literature.
\begin{lemma}[{\cite[Problem 16.25]{csiszárInformationTheoryCoding2011}}]
\label{lem:double_markovity}
    A triple of RVs \((U, X, Y)\) satisfies the conditions
    \begin{equation*}
        U \to X \to Y \quad \text{and} \quad X \to U \to Y
    \end{equation*}
    if and only if there exist functions \(f\) of \(U\) and \(g\) of \(X\) such that
    \begin{enumerate}
        \item \(f(U) = g(X)\) with probability \(1\);
        \item \(Y\) is conditionally independent of \((U,X)\) given \((f(U), g(X))\).
    \end{enumerate}
\end{lemma}
Recall the assumption \(|\mathcal{X}| = 2\) and let \(\mathcal{X} = \{0,1\}\) without loss of generality.
Since a function \(g\) on \(\mathcal{X}\) can take at most two distinct values, we assume without loss of generality that \(g:\mathcal X\to \{0,1\}\).
Then, every \(g\) belongs to one of the following three categories:
\begin{enumerate}
    \item \(g(0) = g(1) = c\) for a constant \(c \in \{0, 1\}\).
    \item \(g(x) = x\).
    \item \(g(x) = x \oplus 1\), where \(\oplus\) denotes the modulo-2 sum.
\end{enumerate}
We now apply Lemma \ref{lem:double_markovity} to the binary \(X\).
We first see that the function \(g\) cannot belong to category \(1\), since in this case \(f(U) = g(X) = c\) and the second condition in Lemma \ref{lem:double_markovity} does not hold.
Hence, \(g\) must belong to category \(2\) or \(3\), meaning that we have either \(X = f(U)\) or \(X = f(U) \oplus 1\).
Thus, \eqref{eq:apd:proof_disproving_Han_conjecture:double_markovity:double_chains} holds only if there exists an onto function \(f^{\prime}\) such that \(X = f^{\prime}(U)\).

We next study the case \(\mathcal{U}_1 \neq \emptyset\).
First, based on the two subsets \(\mathcal{U}_1\) and \(\mathcal{U}_2\), we compute the two conditional distributions \(P_{U | U \in \mathcal{U}_1}\) and \(P_{U|U \in \mathcal{U}_2}\).
Let \(U_1 \sim P_{U | U \in \mathcal{U}_1}\) and \(U_2 \sim P_{U | U \in \mathcal{U}_2}\). Then, we can write
\begin{equation}
    P_{U} = P_{U}[\mathcal{U}_1] \times P_{U_1} + P_{U}[\mathcal{U}_2] \times P_{U_2}.
\end{equation}
We now take a closer look at \(u \in \mathcal{U}_1\).
Recall that if \(u \in \mathcal{U}_1\), we have \(P_{Y|U=u} = P_{Y}\), i.e.,
\begin{equation}
    \sum_{x}\big(P_{X|U}(x|u)-P_{X}(x)\big)P_{Y|X}(y|x)=0. \label{eq:apd:disproving_Han_conjecture:sum_y_marginal_combine_for_U_1}
\end{equation}
Moreover, recall that we excluded the instance \(P_{XY}\) such that \(P_{Y|X} = P_{Y}\), i.e., here we have \(P_{Y|X}(y|0) \neq P_{Y|X}(y|1)\).
Thus, by the same argument used to analyze \eqref{eq:apd:disproving_Han_conjecture:sum_y_marginal_combine}, and in particular \eqref{eq:apd:disproving_Han_conjecture:a_b_equation_1} and \eqref{eq:apd:disproving_Han_conjecture:a_b_equation_2},  we see that \eqref{eq:apd:disproving_Han_conjecture:sum_y_marginal_combine_for_U_1} implies \(P_{X|U=u} = P_{X}\) if \(u \in \mathcal{U}_1\).
It follows that
\begin{align}
    P_{X}(x) & = \sum_{u \in \mathcal{U}}P_{U}(u) P_{X|U}(x|u) \\
    & = \sum_{u \in \mathcal{U}_1}P_{U}(u) P_{X|U}(x|u) + \sum_{u \in \mathcal{U}_2}P_{U}(u) P_{X|U}(x|u) \\
    & = P_{U}[\mathcal{U}_1] P_{X}(x) + \sum_{u \in \mathcal{U}_2}P_{U}(u) P_{X|U}(x|u). 
\end{align}
Since \(1 = P_{U}[\mathcal{U}_1]+ P_{U}[\mathcal{U}_2]\), we get
\begin{align}
    P_{U}[\mathcal{U}_2] P_{X}(x) &= \sum_{u \in \mathcal{U}_2}P_{U}(u) P_{X|U}(x|u), \\
    P_{X}(x) &= \sum_{u \in \mathcal{U}_2}P_{U_2}(u) P_{X|U}(x|u).  \label{eq:apd:disproving_Han_conjecture:marginal_P_X}
\end{align}
Let \(P_{X|U_2} \triangleq P_{X|U}\) and \(P_{U_2X} \triangleq P_{U_2}P_{X|U_2}\).
Then, \eqref{eq:apd:disproving_Han_conjecture:marginal_P_X} implies that \(P_{U_2X} = P_{U_2|X}P_{X}\) for some \(P_{U_2|X}\).
Hence, under \(P_{U_2XY} = P_{U_2|X}P_{XY}\), we also have \(P_{X|U_2} = P_{X|U_2, Y}\) due to \(P_{X|U_2}=P_{X|U}\), i.e., the following two Markov chains hold simultaneously
\begin{equation}
    U_2 \to X \to Y \quad \text{and} \quad X \to U_2 \to Y.
\end{equation}
Therefore, there exists an onto function \(f^{\prime}\) such that \(X = f^{\prime}(U_2)\).

In summary, \eqref{eq:apd:disproving_Han_conjecture:sum_y_marginal_combine} implies that based on the maximizer \(P_{U|X}\), we can  partition \(\mathcal{U}\) into two subsets \(\mathcal{U}_1\) and \(\mathcal{U}_2\) such that
\begin{enumerate}
    \item if \(u \in \mathcal{U}_1\), then \(P_{Y|U=u} = P_{Y}\) and \(P_{X|U=u} = P_{X}\);
    \item there exists an onto function \(f^{\prime}\) such that \(X = f^{\prime}(U_2)\), where \(U_2 \sim P_{U|U\in \mathcal{U}_2}\).
\end{enumerate}

\subsection{Linear Expression}
So far, we have shown that \(E_{\mathrm{H}}(R) = L(R)\) implies \eqref{eq:apd:disproving_Han_conjecture:sum_y_marginal_combine}, which in turn implies the two conditions on their common maximizer \(P_{U|X}\).
Let the maximizer be \(P_{U^{\ast} | X}\).
It follows that \(E_{\mathrm{H}}(R) = L(R)\) only if
\begin{align}
    L(R) & = \max_{\substack{P_{U|X}: \\ I(U;X) \leq R } } L(U;Y) \\
    & = D(P_{U^{\ast}}P_{Y} \| P_{U^{\ast}Y}) \\
    & = \sum_{u\in \mathcal{U}_2, y \in \mathcal{Y}}  P_{U^{\ast}}(u)P_{Y}(y)\log \frac{P_{Y}(y)}{P_{Y|U^{\ast}}(y|u)} \label{eq:apd:proof_disproving_Han_conjecture:double_markovity:linear_expression:omit_U_1} \\
    & = P_{U^{\ast}}[\mathcal{U}_2] \sum_{u\in \mathcal{U}_2, y \in \mathcal{Y}} P_{U^{\ast}_2}(u)P_{Y}(y)\log \frac{P_{Y}(y)}{P_{Y|U^{\ast}_2}(y|u)} \label{eq:apd:proof_disproving_Han_conjecture:double_markovity:linear_expression:conditioned_U_2}  \\
    & = P_{U^{\ast}}[\mathcal{U}_2] \sum_{u\in \mathcal{U}_2, y \in \mathcal{Y}} P_{U^{\ast}_2}(u)P_{Y}(y)\log \frac{P_{Y}(y)}{P_{Y|X}(y|f^{\prime}(u))} \label{eq:apd:proof_disproving_Han_conjecture:double_markovity:linear_expression:linear_mapping_1} \\
    & = P_{U^{\ast}}[\mathcal{U}_2]\sum_{x\in \mathcal{X}, y \in \mathcal{Y}}  P_{X}(x)P_{Y}(y)\log \frac{P_{Y}(y)}{P_{Y|X}(y|x)}  \label{eq:apd:proof_disproving_Han_conjecture:double_markovity:linear_expression:linear_mapping_2} \\
    & = P_{U^{\ast}}[\mathcal{U}_2] D(P_{X}P_{Y} \| P_{XY}),
\end{align}
where \eqref{eq:apd:proof_disproving_Han_conjecture:double_markovity:linear_expression:omit_U_1} is because \(P_{Y|U = u} = P_{Y}\) if \(u \in \mathcal{U}_1\);
in \eqref{eq:apd:proof_disproving_Han_conjecture:double_markovity:linear_expression:conditioned_U_2} we consider the conditional distribution
\begin{equation}
    P_{U^{\ast}_2}(u) \triangleq \frac{P_{U^{\ast}}(u)}{ P_{U^{\ast}}[\mathcal{U}_2] };
\end{equation}
\eqref{eq:apd:proof_disproving_Han_conjecture:double_markovity:linear_expression:linear_mapping_1} and \eqref{eq:apd:proof_disproving_Han_conjecture:double_markovity:linear_expression:linear_mapping_2} are because \(X = f^{\prime}(U^{\ast}_2)\).
On the other hand, \(P_{U^{\ast}|X}\) must also satisfy
\begin{align}
    R & \geq I(U^{\ast};X) \\
    & = H(X)- H(X|U^{\ast}) \\
    & = \sum_{u \in \mathcal{U}} P_{U^{\ast}} (u) \big( H(X) - H(X|U^{\ast} = u) \big) \\
    & = \sum_{u \in \mathcal{U}_2}P_{U^{\ast}} (u) \big( H(X) - H(X|U^{\ast} = u) \big) \label{eq:apd:proof_disproving_Han_conjecture:double_markovity:linear_expression:u_1_X_independence}\\
    & = P_{U^{\ast}}[\mathcal{U}_2] \times H(X), \label{eq:apd:proof_disproving_Han_conjecture:double_markovity:linear_expression:H(X)}
\end{align}
where \eqref{eq:apd:proof_disproving_Han_conjecture:double_markovity:linear_expression:u_1_X_independence} is because \(P_{X|U=u} = P_{X}\)  if \(u \in \mathcal{U}_1\);
\eqref{eq:apd:proof_disproving_Han_conjecture:double_markovity:linear_expression:H(X)} is because \(X = f^{\prime}(U_2)\).
It follows that \(P_{U^{\ast}}[\mathcal{U}_2] = \frac{R}{H(X)}\), and hence \(E_{\mathrm{H}}(R) = L(R)\) implies
\begin{equation}
    E_{\mathrm{H}}(R) = L(R) = \frac{R}{H(X)}  D(P_{X}P_{Y} \| P_{XY}). \label{eq:apd:proof_disproving_Han_conjecture:double_markovity:linear_expression:conclusion}
\end{equation}

\subsection{Example}
\label{sec:apd:proof_disproving_Han_conjecture:example}

\begin{figure}[b]
        \centering
        \scalebox{0.88}{\input{fig/Counterexample.tex}}
        \captionsetup{justification=centering}
        \caption{Example}
        \label{fig:apd:proof_disproving_Han_conjecture:counterexample}
\end{figure}

We now provide an example where \(L(R) > \frac{R}{H(X)}  D(P_{X}P_{Y} \| P_{XY}) \) for every \(0<R<H(X)\), which implies \(E_{\mathrm{H}}(R) > L(R)\) for every \(0<R<H(X)\).
Consider the binary symmetric double source (BSDS)
\begin{equation}
    P_{XY} =
    \begin{bmatrix}
        \frac{1-p}{2} &\frac{p}{2} \\
        \frac{p}{2}   &\frac{1-p}{2}
    \end{bmatrix},
\end{equation}
where \(p=\frac{2-\sqrt{2}}{4}\).
Since we are only interested in lower bounding \(L(R)\), it suffices to restrict the cardinality of \(\mathcal{U}\) in \(L(R)\) from \(|\mathcal{X}|+1\) to \(|\mathcal{X}|\). 
We thus choose \(\mathcal{U} = \{0,1\}\) without loss of generality.
As illustrated in Fig. \ref{fig:apd:proof_disproving_Han_conjecture:counterexample}, we consider a \(P_{U|X}\) such that \(P_{X|U}\) is a binary symmetric channel (BSC) with parameter \(0 < q < \frac{1}{2}\).
It can be seen that, in this example, \(P_{Y|U}\) is another BSC with parameter \(q \ast p\), where \(q \ast p \triangleq q(1-p) + (1-q)p\).
As a result, we have
\begin{align}
    L(U;Y) & = D(P_{U}P_{Y}\|P_{UY}) \\
    & = -\frac{1}{2}\log 4(q \ast p)(1-q\ast p) \\
    & = - \frac{1}{2} \log \big(1-\frac{1}{2}(1-2q)^{2}\big),
\end{align}
and \(I(U;X) = \log 2 -H_{\mathrm{b}}(q)\), where \(H_{\mathrm{b}}(q) \triangleq -q\log q - (1-q) \log (1-q)\).
Let us define
\begin{equation}
    \hat{L}(R)  \triangleq \max_{q: \log 2 - H_{\mathrm{b}}(q) \leq R} - \frac{1}{2} \log \big(1-\frac{1}{2}(1-2q)^{2}\big).
\end{equation}
It is clear that \(L(R) \geq \hat{L}(R)\), because the maximization defining  \(L(R)\) includes the BSC auxiliaries considered here.
We next show that, in this example, we have \(\hat{L}(R) > \frac{R}{H(X)} D(P_{X}P_{Y}\|P_{XY})\) for every \(0<R<H(X)\), which then implies \(L(R) > \frac{R}{H(X)} D(P_{X}P_{Y}\|P_{XY})\) for every \(0<R<H(X)\).

Let \( \log 2 - H_{\mathrm{b}}(q_0) = R\) where \(0<q_0<\frac{1}{2}\).
It can be verified that we have
\begin{align}
    \hat{L}(R) = - \frac{1}{2} \log \big(1-\frac{1}{2}(1-2q_0)^{2}\big)
\end{align}
On the other hand, we have
\begin{equation}
    D(P_{X}P_{Y} \| P_{XY}) = \frac{1}{2}\log 2.
\end{equation}
Since \(H(X)= \log 2\),  to prove \(\hat{L}(R) > \frac{R}{H(X)} D(P_{X}P_{Y}\|P_{XY})\) is equivalent to proving \(\hat{L}(R) > \frac{1}{2}R\), which in turn is equivalent to
\begin{equation}
    - \log \big(1-\frac{1}{2}(1-2q_0)^{2}\big) + H_{b}(q_0) - \log 2 > 0.\label{eq:apd:proof_disproving_Han_conjecture:concavity_example:strict_greater}
\end{equation}
Let \(t \triangleq 1-2q_0\), and define
\begin{equation}
    f(t) = - \log \big(1-\frac{1}{2}t^{2}\big) + H_{b}(\frac{1-t}{2}) - \log 2
\end{equation}
By taking the derivatives, we obtain that
\begin{align}
    f^{\prime}(t) & = \frac{t}{1-\frac{t^2}{2}} - \frac{1}{2}\log \frac{1+t}{1-t},\\
     f^{\prime \prime}(t) & = \frac{t^2(2-3t^2)}{(1-t^2)(2-t^2)^2}.
\end{align}
Thus, \(f^{\prime \prime}(t) > 0\) for \( 0< t < \sqrt{\frac{2}{3}}\), and \(f^{\prime \prime}(t) < 0\) for \(  \sqrt{\frac{2}{3}} < t< 1\).
Since \(f^{\prime}(0) = 0 \) and \(f^{\prime}(t) \to - \infty \) as \(t \to 1\), we see that \(f(t)\) first increases and then decreases.
Together with \(f(0) = f(1) = 0\), we see that \(f(t)>0\) for any \(0<t<1\).
In summary, \eqref{eq:apd:proof_disproving_Han_conjecture:concavity_example:strict_greater} holds for all \( 0 < q_0 < \frac{1}{2}\), meaning that \(\hat{L}(R) > \frac{R}{H(X)} D(P_{X}P_{Y}\|P_{XY})\) holds for every \(0<R<H(X)\), which completes the proof.

\begin{remark}
    We can show that \(L(R)\) is concave in \(R\). 
    Consider \(0 \leq R_1 \leq R_2 \leq H(X)\) and \(\theta \in (0,1)\).
    Suppose the maximizations in \(L(R_1)\) and \(L(R_2)\) are achieved by \(P_{U_1|X}\) and \(P_{U_2 | X}\) respectively.
    Let \(J\) be a random variable with \(P_{J}(j=1) = \theta\) and \(P_{J}(j=2) = 1- \theta\).
    Then, we have
    \begin{align}
        I(U_J, J;X) & = I(J;X) + I(U_J;X|J) \\
        & = \theta I(U_1;X) + (1-\theta) I(U_2;X) \\
        & \leq \theta R_1 + (1-\theta) R_2. \label{eq:apd:proof_disproving_Han_conjecture:concavity_example:concavity_1}
    \end{align}
    Similarly, we have
    \begin{equation}
        L(U_J,J;Y) = \theta L(U_1;Y) + (1-\theta) L(U_2;Y).
    \end{equation}
    It follows that
    \begin{align}
        \theta L(R_1) + (1-\theta) L(R_2) & = \theta L(U_1;Y) + (1-\theta) L(U_2;Y) \\
        & = L(U_J,J;Y) \\
        & \leq L( \theta R_1 + (1-\theta) R_2 ) \label{eq:apd:proof_disproving_Han_conjecture:concavity_example:concavity_2},
    \end{align}
    where \eqref{eq:apd:proof_disproving_Han_conjecture:concavity_example:concavity_2} follows from \eqref{eq:apd:proof_disproving_Han_conjecture:concavity_example:concavity_1}.
    Hence, in the binary case, unless \(L(R)\) coincides with the time-sharing line, i.e.,
    \begin{equation}
        L(R) = \frac{R}{H(X)}  D(P_{X}P_{Y} \| P_{XY}) \qquad \forall R \leq H(X),
    \end{equation}
    we always have \(E_{\mathrm{H}}(R) > L(R)\) for every \(0<R<H(X)\).
\end{remark}

\section{Converse Proof for Testing Against Independence}
\label{apd:converse_proof_generalization_testing_against_independence}

In this appendix, we show that \(E(R) = E_{\mathrm{H}}(R)\) under the more general formulation of testing against independence, where the receiver decides between
    \begin{align}
        \mathcal{H}_0: (X, Y) & \sim P_{XY}, \\
        \mathcal{H}_1: (X, Y) & \sim Q_{X}Q_{Y}.
    \end{align}
This result may be known to experts in the field. However, as we are not aware of a proof in the literature, we provide one here for completeness.
Similar to the proof of Theorem \ref{thm:testing_against_dependence}, we only need to establish the converse part \(E(R) \leq E_{\mathrm{H}}(R)\).
We achieve this by extending the converse proof of the special case when the receiver decides between \(P_{XY}\) and \(P_{X}P_{Y}\).

We begin with the multi-letter characterization of the achievable exponent.
Consider a sequence of testing schemes \((f_n, \varphi_n)\) with \(||f_n|| \ndot{\leq} e^{nR}\), and denote by \(E\) the achievable error exponent under the sequence.
By applying the schemes to the two hypotheses respectively, we obtain two joint distributions \(P_{MX^nY^n}\) and \(Q_{MX^n}Q_{Y^n}\), where \(Q_{Y^n} = Q_{Y}^n\).
Following the multi-letter characterization proof \cite[Theorem 1]{ahlswedeHypothesisTestingCommunication1986}, we see that
\begin{equation}
     E \leq \liminf_{n\to \infty} \frac{1}{n} D(P_{MY^n} \| Q_{M}Q_{Y^n}).
\end{equation}
Notice that
\begin{align}
    & \frac{1}{n} D(P_{MY^n} \| Q_{M}Q_{Y^n}) \nonumber \\
    & = \frac{1}{n}D(P_{M} \| Q_{M}) + \frac{1}{n} D(P_{Y^n} \| Q_{Y^n}) + \frac{1}{n} D(P_{MY^n} \| P_{M}P_{Y^n}) \\
    & \leq \frac{1}{n}D(P_{X^n} \| Q_{X^n}) + \frac{1}{n} D(P_{Y^n} \| Q_{Y^n}) + \frac{1}{n} D(P_{MY^n} \| P_{M}P_{Y^n}) \label{eq:apd:converse_proof_generalization_testing_against_independence:data_processing_inequality_M_to_X^n} \\
    & = D(P_{X} \| Q_{X}) + D(P_{Y}\|Q_{Y}) + \frac{1}{n} I(M;Y^n). \label{eq:apd:converse_proof_generalization_testing_against_independence:M_Y^n}
\end{align}
where \eqref{eq:apd:converse_proof_generalization_testing_against_independence:data_processing_inequality_M_to_X^n} is due to the data processing inequality.
Define \(U_i \triangleq (M, X^{i-1})\), and then \(\bar{U} \triangleq (U_J, J)\) where \(J\)  is uniformly distributed over \([n]\).
Following the weak converse technique \cite{wynerSourceCodingSide1975,ahlswedeSourceCodingSide1975,gamalNetworkInformationTheory2011}, we can show that
\begin{equation}
    \frac{1}{n}I(M; Y^n) \leq I(U_J,J; Y_J) = I(P_{Y},P_{\bar{U}|Y}) = D(P_{\bar{U}Y} \| P_{\bar{U}}P_{Y}), \label{eq:apd:converse_proof_generalization_testing_against_independence:weak_converse_single_letterization_Y_U}
\end{equation}
and
\begin{equation}
    R \geq I(U_J,J;X_J) = I(\bar{U};X), \label{eq:apd:converse_proof_generalization_testing_against_independence:weak_converse_single_letterization_X_U}
\end{equation}
where we have \(P_{\bar{U}XY} = P_{\bar{U}|X}P_{XY}\).

We next upper bound \eqref{eq:apd:converse_proof_generalization_testing_against_independence:weak_converse_single_letterization_Y_U} as follows.
Consider a distribution \(\tilde{P}_{\bar{U}XY}\) that satisfies
\begin{align}
    \tilde{P}_{\bar{U}X} & = P_{\bar{U}X}, \label{eq:apd:converse_proof_generalization_testing_against_independence:marginal_condition_1} \\
    \tilde{P}_{\bar{U}Y} & = P_{\bar{U}Y}. \label{eq:apd:converse_proof_generalization_testing_against_independence:marginal_condition_2}
\end{align}
From the data processing inequality, we must have
\begin{equation}
    D(P_{\bar{U}Y} \| P_{\bar{U}}P_{Y}) \leq D(\tilde{P}_{\bar{U}XY} \| P_{\bar{U}X}P_{Y} ), \label{eq:apd:converse_proof_generalization_testing_against_independence:data_processing_adding_U}
\end{equation}
where the processing corresponds to marginalization.
Substituting \eqref{eq:apd:converse_proof_generalization_testing_against_independence:weak_converse_single_letterization_Y_U} and \eqref{eq:apd:converse_proof_generalization_testing_against_independence:data_processing_adding_U} into \eqref{eq:apd:converse_proof_generalization_testing_against_independence:M_Y^n} yields that
\begin{align}
    \frac{1}{n} D(P_{MY^n} \| Q_{M}Q_{Y^n}) &\ndot{\leq} D(P_{X} \| Q_{X}) + D(P_{Y}\|Q_{Y}) +  D(\tilde{P}_{\bar{U}XY} \| P_{\bar{U}X}P_{Y} ) \\
    & = D(\tilde{P}_{\bar{U}XY} \| P_{\bar{U}|X}Q_XQ_{Y} ) \label{eq:apd:converse_proof_generalization_testing_against_independence:data_processing_P_tilde}\\
    & = \min_{\tilde{P}_{\bar{U}XY}} D(\tilde{P}_{\bar{U}XY} \| P_{\bar{U}|X}Q_XQ_{Y} )  \label{eq:apd:converse_proof_generalization_testing_against_independence:data_processing_minimizing_over_P_tilde} \\
    & = E(P_{\bar{U}XY}\| Q_{\bar{U}X}Q_{Y}) \\
    & \leq \max_{ \substack{P_{U|X}: \\ I(U;X) \leq R} }  E(P_{UXY}\| Q_{UX}Q_{Y}),  \label{eq:apd:converse_proof_generalization_testing_against_independence:maximize_over_U}
\end{align}
where \eqref{eq:apd:converse_proof_generalization_testing_against_independence:data_processing_minimizing_over_P_tilde} is because \eqref{eq:apd:converse_proof_generalization_testing_against_independence:data_processing_P_tilde} holds for any \(\tilde{P}_{\bar{U}XY}\) satisfying \eqref{eq:apd:converse_proof_generalization_testing_against_independence:marginal_condition_1} and \eqref{eq:apd:converse_proof_generalization_testing_against_independence:marginal_condition_2}, and hence we can minimize over such \(\tilde{P}_{\bar{U}XY}\) to obtain the tightest upper bound; \eqref{eq:apd:converse_proof_generalization_testing_against_independence:maximize_over_U} follows from maximizing over all \(U\) satisfying \(I(U;X) \leq R\).
The cardinality bound \(|\mathcal{U}| \leq |\mathcal{X}| + 1\) is due to Han \cite[Theorem 3]{hanHypothesisTestingMultiterminal1987}.
With this, the converse proof is completed.

\begin{remark}
    We can readily derive an alternative single-letter upper bound from \eqref{eq:apd:converse_proof_generalization_testing_against_independence:M_Y^n}.
    Substituting \eqref{eq:apd:converse_proof_generalization_testing_against_independence:weak_converse_single_letterization_Y_U} into \eqref{eq:apd:converse_proof_generalization_testing_against_independence:M_Y^n}, and considering \eqref{eq:apd:converse_proof_generalization_testing_against_independence:weak_converse_single_letterization_X_U}, we can deduce that
    \begin{align}
        \frac{1}{n} D(P_{MY^n} \| Q_{M}Q_{Y^n}) & \leq D(P_{X} \| Q_{X}) + D(P_{Y}\|Q_{Y}) +  \max_{ \substack{P_{U|X}: \\ I(U;X) \leq R} } I(U;Y) \\
        & = \max_{ \substack{P_{U|X}: \\ I(U;X) \leq R} } D(P_{X} \| Q_{X}) + D(P_{Y}\|Q_{Y}) + I(U;Y), \label{eq:apd:converse_proof_generalization_testing_against_independence:simplified_expression}
    \end{align}
    where \eqref{eq:apd:converse_proof_generalization_testing_against_independence:simplified_expression} is because \(D(P_{X} \| Q_{X}) + D(P_{Y}\|Q_{Y})\) is constant.
    Using arguments similar to  those in \eqref{eq:apd:converse_proof_generalization_testing_against_independence:data_processing_adding_U} and \eqref{eq:apd:converse_proof_generalization_testing_against_independence:data_processing_minimizing_over_P_tilde}, we see that the RHS of \eqref{eq:apd:converse_proof_generalization_testing_against_independence:simplified_expression} is no greater than Han's exponent, and hence the latter reduces to the former in this case.
\end{remark}

\section{Proof of Equality of the Two Expressions}
\label{apd:proof_equality_conditional_and_Han_two_exponents}
In this appendix, we show that 
\begin{equation}
     \max_{\substack{P_{U|XZ}: \\ I(U;X|Z) \leq R }} \tilde{E}(P_{UXZ}P_{Y|Z} \| Q_{UXYZ}) = \max_{\substack{P_{V|XZ}: \\ I(V;X, Z) \leq R }} \tilde{E}(P_{VXZ}P_{Y|Z} \| Q_{VXYZ}).\label{eq:apd_equality_two_exponents:equality_to_prove}
\end{equation}
Since the LHS of \eqref{eq:apd_equality_two_exponents:equality_to_prove} is already known to be at least as large as the RHS, it remains to prove the reverse inequality, i.e., the LHS is no larger than the RHS.
To this end, let \(P_{U|XZ}^{\ast}\) be an optimal solution to the LHS of \eqref{eq:apd_equality_two_exponents:equality_to_prove}.
From \(P_{XZ}P_{U|XZ}^{\ast}\), we obtain the conditional distribution \(P_{U|Z}^{\ast}\).
For simplicity, we assume \(\mathcal{Z} = \{0,1\}\) in what follows. 
The same argument extends to an arbitrary \(\mathcal{Z}\).

Consider a vector \(V \triangleq ( U_0, U_1 )\), where both \(U_0\) and \(U_1\) are supported on \(\mathcal{U}\).
We now define the conditional distribution \(P_{V|XZ}\) such that 
\begin{equation}
    P_{V|X, Z=0} = P_{U_0|X, Z=0}P_{U_1|Z=0}, \label{eq:apd_equality_two_exponents:definition_V_X_Z_0}
\end{equation}
where \(P_{U_0|X,Z=0} \triangleq P_{U|X, Z=0}^{\ast} \) and \( P_{U_1|Z=0} \triangleq P_{U|Z=1}^{\ast}\); and
\begin{equation}
    P_{V| X, Z=1} = P_{U_0|Z=1}P_{U_1|X, Z=1},\label{eq:apd_equality_two_exponents:definition_V_X_Z_1}
\end{equation}
where \(P_{U_0| Z=1} \triangleq P_{U| Z=0}^{\ast} \) and \( P_{U_1|X, Z=1} \triangleq P_{U|X, Z=1}^{\ast}\).
It can be verified that under \(P_{XZ}P_{V|XZ}\) we have
\begin{equation}
    P_{V|Z=0} = P_{V|Z=1} = P_{U|Z=0}^{\ast}P_{U|Z=1}^{\ast}, \label{eq:apd_equality_two_exponents:V_Z_independence}
\end{equation}
i.e., \(V\) and \(Z\) are independent.
It follows that
\begin{align}
    I(V;X,Z) & = I(V; X |Z) \\
    & = H(V|Z) - H(V|X,Z) \\
    & = \sum_{i \in \{0,1\}}P_{Z}(i)\big(H(V|Z=i) - H(V|X, Z=i)\big)  \\
    & = \sum_{i \in \{0,1\}}P_{Z}(i)\big(H(U^{\ast}|Z=0) + H(U^{\ast}|Z=1) - H(V|X, Z=i)\big) \label{eq:apd_equality_two_exponents:using_V_Z_independence}  \\
    & = \sum_{i \in \{0,1\}}P_{Z}(i)\big(H(U^{\ast}|Z=0) + H(U^{\ast}|Z=1) - H(U^{\ast}|X, Z=i) - H(U^{\ast}|Z=i\oplus1)\big) \label{eq:apd_equality_two_exponents:using_definition_V_X_Z}  \\
    & = \sum_{i \in \{0,1\}}P_{Z}(i)\big(H(U^{\ast}|Z=i)  - H(U^{\ast}|X, Z=i)\big) \label{eq:apd_equality_two_exponents:verification}\\
    & = I(U^{\ast};X|Z),
\end{align}
where \eqref{eq:apd_equality_two_exponents:using_V_Z_independence} is due to \eqref{eq:apd_equality_two_exponents:V_Z_independence};
in \eqref{eq:apd_equality_two_exponents:using_definition_V_X_Z}, we use \(\oplus\) to denote the modulo-\(2\) sum,  and \eqref{eq:apd_equality_two_exponents:using_definition_V_X_Z}  follows from the definition of \(P_{V|XZ}\) in \eqref{eq:apd_equality_two_exponents:definition_V_X_Z_0} and \eqref{eq:apd_equality_two_exponents:definition_V_X_Z_1};
\eqref{eq:apd_equality_two_exponents:verification} can be verified by considering \(i=0\) and \(1\) respectively.
Thus, we see that \(I(V;X,Z) = I(U^{\ast};X|Z) \leq R\).
On the other hand, from the data processing inequality, we obtain that
\begin{equation}
    \tilde{E}(P_{U^{\ast}XZ}P_{Y|Z} \| Q_{U^{\ast}XYZ}) \leq \tilde{E}(P_{VXZ}P_{Y|Z} \| Q_{VXYZ}),
\end{equation}
since \(V = (U_0, U_1)\) and for each value of \(Z\), one component of \(V\) coincides with \(U^{\ast}\).
Therefore, the LHS of \eqref{eq:apd_equality_two_exponents:equality_to_prove} is no larger than the RHS, which completes the proof.
Although the alphabet of \(V\) is enlarged in the proof, this causes no loss of generality, since the cardinality reduction argument can be applied afterward.

\section{Proof of the Cardinality Bound in Theorem \ref{thm:conditional_testing_against_dependence}}
\label{apd:cardinality_bound_conditional_testing_against_dependence}

Let \(P_{UXYZ} = P_{UXZ}P_{Y|Z}\), and recall that \(Q_{UXYZ} = P_{U|XZ}Q_{XYZ}\).  In this appendix, we show that
\begin{equation}
    \tilde{E}(R) = \max_{\substack{P_{U|XZ}: \\ I(U;X|Z) \leq R }} \tilde{E}(P_{UXYZ} \| Q_{UXYZ})
\end{equation}
remains unchanged if we impose the cardinality constraint \(|\mathcal{U}| \leq |\mathcal{X}||\mathcal{Z}| + 1\).
We accomplish this by showing that for any feasible \(P_{U|XZ}\) with \(|\mathcal{U}| > |\mathcal{X}||\mathcal{Z}| + 1\), we can find a \(P_{U^{\prime} | XZ}\) with \(|\mathcal{U}^{\prime}| \leq |\mathcal{X}||\mathcal{Z}| + 1\) such that \(I(U^{\prime};X|Z) \leq R\) and 
\begin{equation}
    \tilde{E}(P_{UXYZ} \| Q_{UXYZ}) = \tilde{E}(P_{U^{\prime}XYZ} \| Q_{U^{\prime}XYZ}). \label{eq:apd:cardinality_bound_conditional_testing_against_dependence:goal}
\end{equation}

To this end, recall the definition of \(\tilde{E}(P_{UXYZ} \| Q_{UXYZ})\) in \eqref{eq:main_results:conditional_testing_against_dependence:definition_E_tilde} and we see that
\begin{align}
    \tilde{E}(P_{UXYZ} \| Q_{UXYZ}) & = \min_{\substack{\tilde{P}_{UXYZ} \in \tilde{\mathcal{P}}_{\mathrm{H}}(P_{UXYZ})  }} D( \tilde{P}_{UXYZ} \| Q_{UXYZ} ) \label{eq:apd:cardinality_bound_conditional_testing_against_dependence:starting} \\
    & = \min_{\substack{\tilde{P}_{UXYZ} \in \tilde{\mathcal{P}}_{\mathrm{H}}(P_{UXYZ})  }} D( \tilde{P}_{UXYZ} \| P_{U|XZ}Q_{XYZ}) \\
    & =  \min_{\substack{\tilde{P}_{UXYZ} \in \tilde{\mathcal{P}}_{\mathrm{H}}(P_{UXYZ})  }} D( \tilde{P}_{UXYZ} \| P_{U}Q_{XYZ}) - I(U;X,Z) \label{eq:apd:cardinality_bound_conditional_testing_against_dependence:property_P}\\
    & = -H(X,Z) + H(X,Z|U) + \min_{\substack{\tilde{P}_{UXYZ} \in \tilde{\mathcal{P}}_{\mathrm{H}}(P_{UXYZ})  }} D( \tilde{P}_{XYZ|U} \| Q_{XYZ} | P_{U}), 
\end{align}
where \eqref{eq:apd:cardinality_bound_conditional_testing_against_dependence:property_P} is because \(\tilde{P}_{UXYZ}\) satisfies \(\tilde{P}_{UXZ} = P_{UXZ}\).
Recall that
\begin{equation}
     \tilde{\mathcal{P}}_{\mathrm{H}}(P_{UXYZ})  = \{ \tilde{P}_{UXYZ}: \tilde{P}_{UXZ} = P_{UXZ}, \tilde{P}_{UYZ} = P_{UYZ}\}.
\end{equation}
Let us now consider
\begin{equation}
     \tilde{\mathcal{P}}_{\mathrm{H}}(P_{XYZ|U = u})  \triangleq \{ \tilde{P}_{XYZ|U = u}: \tilde{P}_{XZ|U = u} = P_{XZ| U=u}, \tilde{P}_{YZ|U=u} = P_{YZ|U=u}\}.
\end{equation}
It can then be verified that
\begin{align}
    & \min_{\substack{\tilde{P}_{UXYZ} \in \tilde{\mathcal{P}}_{\mathrm{H}}(P_{UXYZ})  }} D( \tilde{P}_{XYZ|U} \| Q_{XYZ} | P_{U}) \nonumber \\
    & = \min_{\substack{\tilde{P}_{UXYZ} \in \tilde{\mathcal{P}}_{\mathrm{H}}(P_{UXYZ})  }} \sum_{u} P_{U}(u) \times D( \tilde{P}_{XYZ|U=u} \| Q_{XYZ}) \\
    & = \sum_{u}P_{U}(u) \times \min_{\substack{\tilde{P}_{XYZ|U=u} \in \tilde{\mathcal{P}}_{\mathrm{H}}(P_{XYZ|U=u})  }} D( \tilde{P}_{XYZ|U=u} \| Q_{XYZ}).  \label{eq:apd:cardinality_bound_conditional_testing_against_dependence:move_P_U_out}
\end{align}
Notice that \(P_{YZ|U = u} = \sum_{x}P_{XYZ|U=u} = \sum_{x}P_{XZ|U=u}P_{Y|XZ}\).
Hence, given a fixed \(P_{XYZ}\), the set \(\tilde{\mathcal{P}}_{\mathrm{H}}(P_{XYZ|U=u}) \) effectively only depends on \(P_{XZ| U=u}\). 
Based on this observation, we define a function of \(P_{XZ|U=u}\) by
\begin{equation}
    g_1(P_{XZ|U=u}) \triangleq H(X,Z|U=u) + \min_{\substack{\tilde{P}_{XYZ|U=u} \in \tilde{\mathcal{P}}_{\mathrm{H}}(P_{XYZ|U=u})  }} D( \tilde{P}_{XYZ|U=u} \| Q_{XYZ}).
\end{equation}
Proceeding similarly to the proof of \cite[Proposition 2]{watanabeSuboptimalityRandomBinning2022}, we can use the Lagrangian to solve the optimization problem defining \(g_1(P_{XZ|U=u})\), and then it can be verified that \(g_1(P_{XZ|U=u})\) is continuous in \(P_{XZ|U=u}\).
Thus, we can write
\begin{equation}
    \tilde{E}(P_{UXYZ} \| Q_{UXYZ}) = -H(X,Z) + \sum_{u}P_{U}(u) \times g_1(P_{XZ|U=u}). \label{eq:apd:cardinality_bound_conditional_testing_against_dependence:g_1_function}
\end{equation}

Consider the following functions
\begin{equation}
    g_{xz}(P_{XZ}) = P_{XZ}(x,z), \qquad \forall x, z
\end{equation}
and
\begin{equation}
    g_0(P_{XZ}) = H(X|Z).
\end{equation}
It follows that
\begin{align}
    \sum_{u}P_{U}(u)g_{xz}(P_{XZ|U=u}) &= P_{XZ}(x,z), \qquad \forall x, z \\
    \sum_{u}P_{U}(u)g_{0}(P_{XZ|U=u}) & = H(X|Z,U).
\end{align}
Therefore, by the support lemma \cite[Appendix C]{gamalNetworkInformationTheory2011}, there exists a \(P_{U^{\prime}XZ}\) with \(|\mathcal{U}^{\prime}| \leq |\mathcal{X}||\mathcal{Z}| + 1\) such that
\begin{align}
    \sum_{u^{\prime}}P_{U^{\prime}}(u^{\prime})g_{xz}(P_{XZ|U^{\prime}=u^{\prime}}) &= P_{XZ}(x,z), \qquad \forall x, z \label{eq:apd:cardinality_bound_conditional_testing_against_dependence:support_lemma_1}  \\
    \sum_{u^{\prime}}P_{U^{\prime}}(u^{\prime})g_{0}(P_{XZ|U^{\prime}=u^{\prime}}) & = H(X|Z,U), \label{eq:apd:cardinality_bound_conditional_testing_against_dependence:support_lemma_2} \\
    \sum_{u^{\prime}}P_{U^{\prime}}(u^{\prime}) \times g_1(P_{XZ|U^{\prime}=u^{\prime}}) & = \sum_{u}P_{U}(u) \times g_1(P_{XZ|U=u}). \label{eq:apd:cardinality_bound_conditional_testing_against_dependence:support_lemma_3} 
\end{align}
From \eqref{eq:apd:cardinality_bound_conditional_testing_against_dependence:support_lemma_1}, we see that \(P_{U^{\prime}XZ}\) satisfies the marginal distribution constraint \(P_{U^{\prime}XZ} = P_{U^{\prime}|XZ}P_{XZ}\).
Together with \eqref{eq:apd:cardinality_bound_conditional_testing_against_dependence:support_lemma_2}, it follows that \(P_{U^{\prime} | XZ}\) satisfies \(I(U^{\prime};X|Z) = I(U;X| Z) \leq R\).
Recall our goal in \eqref{eq:apd:cardinality_bound_conditional_testing_against_dependence:goal}. By \eqref{eq:apd:cardinality_bound_conditional_testing_against_dependence:support_lemma_3}, we now have
\begin{align}
    \tilde{E}(P_{UXYZ} \| Q_{UXYZ}) & = -H(X,Z) + \sum_{u}P_{U}(u) \times g_1(P_{XZ|U=u}) \\
    & = -H(X,Z) + \sum_{u^{\prime}}P_{U^{\prime}}(u^{\prime}) \times g_1(P_{XZ|U^{\prime}=u^{\prime}}) \\
    & = \tilde{E}(P_{U^{\prime}XYZ} \| Q_{U^{\prime}XYZ}), \label{eq:apd:cardinality_bound_conditional_testing_against_dependence:final}
\end{align}
where \eqref{eq:apd:cardinality_bound_conditional_testing_against_dependence:final} follows from reversing the derivation from \eqref{eq:apd:cardinality_bound_conditional_testing_against_dependence:starting} to \eqref{eq:apd:cardinality_bound_conditional_testing_against_dependence:g_1_function}.
With this, the proof is completed.

\begin{remark}
    \label{rem:cardinality_bound}
    The key difference between our proof and Han's proof lies in \eqref{eq:apd:cardinality_bound_conditional_testing_against_dependence:g_1_function}, where we observe that \(\tilde{E}(P_{UXYZ} \| Q_{UXYZ})\) can be written as a linear function of \(P_{U}\).
    If we followed Han's proof, we would begin by using the Lagrangian to solve \eqref{eq:apd:cardinality_bound_conditional_testing_against_dependence:starting}, see \cite[Lemma 5]{hanHypothesisTestingMultiterminal1987}. 
    Since different \(P_{UXYZ}\) require different optimal Lagrange multipliers, we could not simply reconstruct \(\tilde{E}(P_{U^{\prime}XYZ} \| Q_{U^{\prime}XYZ})\)  as we did in \eqref{eq:apd:cardinality_bound_conditional_testing_against_dependence:final}, because the multipliers optimal for \(P_{UXYZ}\) may not be optimal for \(P_{U^{\prime}XYZ}\).
    Han uses a sophisticated argument to circumvent this issue.
    In our proof, we first observe \eqref{eq:apd:cardinality_bound_conditional_testing_against_dependence:move_P_U_out}, which yields \eqref{eq:apd:cardinality_bound_conditional_testing_against_dependence:g_1_function} and then allows us to directly use the support lemma to prove the cardinality bound.
\end{remark}

\section*{Acknowledgment}
\emph{Disclosure of AI Use:} We acknowledge ChatGPT 5.5 for helping to identify a numerical counterexample to Han's conjecture on lautum information, which then motivated this work on Han's exponent for testing against dependence. 
We also acknowledge ChatGPT 6 Astra for proofreading the manuscript and providing feedback that led to the following revisions:
\begin{enumerate}
    \item For conditional testing against dependence, we had initially believed that the exponent derived in Theorem \ref{thm:conditional_testing_against_dependence} was generally superior to those obtained by applying Han's exponent and the SHA exponent, and had included comparisons based on this belief. ChatGPT identified that this claim was incorrect and suggested an argument establishing the equality. After independently verifying and adapting the argument, we included the resulting proof in Appendix \ref{apd:proof_equality_conditional_and_Han_two_exponents} and removed the corresponding comparisons.
    \item In Appendix \ref{sec:apd:proof_disproving_Han_conjecture:example}, we had originally chosen \(p=\frac{1}{4}\) and claimed that \(\hat{L}(R) > \frac{R}{H(X)} D(P_{X}P_{Y}\|P_{XY})\).
    Although the conclusion for \(p=\frac{1}{4}\) is correct, the proof we initially gave contained a numerical error. ChatGPT identified this error and suggested the alternative choice \(p=\frac{2-\sqrt{2}}{4}\) as well as a derivative-based approach. We adopted this choice and approach in the revised manuscript.
    \item ChatGPT also identified typos and provided comments that improved the presentation of the manuscript.
\end{enumerate}
We independently verified all of the AI-generated suggestions described above and take full responsibility for the content of this paper.

\bibliography{ref}
\bibliographystyle{IEEEtran}

\end{document}

%% file: fig/DHT_model.tex
\begin{tikzpicture}
        \node at (0,0) [name = source] { $X^n$ } ;
        \node at (2.5,0) [draw, rectangle,  name = tx, label= above:{\small Transmitter}, minimum height=0.8cm,minimum width=1.8cm] {$f_n$};
        \node at (7,0) [draw, rectangle,  name = rx, label= above:{\small Receiver},  minimum height=0.8cm,minimum width=1.8cm] {$\varphi_n$};
        \node at (10,0) [name = output] { $ \mathcal{H}_0 / \mathcal{H}_1$ };
        \node at (7,-1.5) [name = SI] { $Y^n$ };

        \draw [->] (source.east) -- (tx.west);
        \draw [->] (tx.east) -- node[align = center, above]{\small $M$} (rx.west);
        \draw [->] (rx.east) -- (output.west);
        \draw [->] (SI.north) -- (rx.south);
        
\end{tikzpicture}

%% file: fig/Testing_against_product_dependence_independence.tex
\begin{tikzpicture}
        \node at (0,0) [name = source] { $(X^n_1, X_2^n)$ } ;
        \node at (2.5,0) [draw, rectangle,  name = tx, label= above:{\small Transmitter}, minimum height=0.8cm,minimum width=1.8cm] {$f_n$};
        \node at (7,0) [draw, rectangle,  name = rx, label= above:{\small Receiver},  minimum height=0.8cm,minimum width=1.8cm] {$\varphi_n$};
        \node at (10,0) [name = output] { $ \mathcal{H}_0 / \mathcal{H}_1$ };
        \node at (7,-1.5) [name = SI] { $(Y_1^n, Y_2^n)$ };

        \draw [->] (source.east) -- (tx.west);
        \draw [->] (tx.east) -- node[align = center, above]{\small $M$} (rx.west);
        \draw [->] (rx.east) -- (output.west);
        \draw [->] (SI.north) -- (rx.south);
        
\end{tikzpicture}

%% file: fig/Conditional_testing_against_dependence.tex
\begin{tikzpicture}
        \node at (0,0) [name = source] { $X^n$ } ;
        \node at (2.5,0) [draw, rectangle,  name = tx, label= above:{\small Transmitter}, minimum height=0.8cm,minimum width=1.8cm] {$\tilde{f}_n$};
        \node at (7,0) [draw, rectangle,  name = rx, label= above:{\small Receiver},  minimum height=0.8cm,minimum width=1.8cm] {$\varphi_n$};
        \node at (10,0) [name = output] { $ \mathcal{H}_0 / \mathcal{H}_1$ };
        \node at (7,-1.5) [name = SI] { $(Y^n, Z^n)$ }; 
        \node at (2.5,-1.5) [name = TSI] { $Z^n$ };

        \draw [->] (source.east) -- (tx.west);
        \draw [->] (tx.east) -- node[align = center, above]{\small $M$} (rx.west);
        \draw [->] (rx.east) -- (output.west);
        \draw [->] (SI.north) -- (rx.south);
        \draw [->] (TSI.north) -- (tx.south);
        
\end{tikzpicture}

%% file: fig/Testing_against_conditional_dependence.tex
\begin{tikzpicture}
        \node at (0,0) [name = source] { $X^n$ } ;
        \node at (2.5,0) [draw, rectangle,  name = tx, label= above:{\small Transmitter}, minimum height=0.8cm,minimum width=1.8cm] {$f_n$};
        \node at (7,0) [draw, rectangle,  name = rx, label= above:{\small Receiver},  minimum height=0.8cm,minimum width=1.8cm] {$\varphi_n$};
        \node at (10,0) [name = output] { $ \mathcal{H}_0 / \mathcal{H}_1$ };
        \node at (7,-1.5) [name = SI] { $(Y^n, Z^n)$ };

        \draw [->] (source.east) -- (tx.west);
        \draw [->] (tx.east) -- node[align = center, above]{\small $M$} (rx.west);
        \draw [->] (rx.east) -- (output.west);
        \draw [->] (SI.north) -- (rx.south);

\end{tikzpicture}

%% file: fig/Counterexample.tex
\begin{tikzpicture}[
    >=stealth,
]

\node at (0,1.7) {$U$};
\node at (3.5,1.7) {$X$};
\node at (7,1.7) {$Y$};

\node (u0) at (0,1) {$1$};
\node (u1) at (0,-1) {$0$};

\node (x0) at (3.5,1) {$1$};
\node (x1) at (3.5,-1) {$0$};

\node (y0) at (7,1) {$1$};
\node (y1) at (7,-1) {$0$};

\draw[->] (u0) -- node[above] {$1-q$} (x0);
\draw[->] (u0) -- node[above, sloped, pos=0.45] {$q$} (x1);
\draw[->] (u1) -- node[below, sloped, pos=0.45] {$q$} (x0);
\draw[->] (u1) -- node[below] {$1-q$} (x1);

\draw[->] (x0) -- node[above] {$1-p$} (y0);
\draw[->] (x0) -- node[above, sloped, pos=0.45] {$p$} (y1);
\draw[->] (x1) -- node[below, sloped, pos=0.45] {$p$} (y0);
\draw[->] (x1) -- node[below] {$1-p$} (y1);

\end{tikzpicture}